\pdfoutput=1
\DocumentMetadata{pdfversion=1.7, pdfstandard=A-3u, lang=en-US}
\RequirePackage{cmap}

\documentclass[final, copyright, creativecommons]{eptcs}

\usepackage[T1]{fontenc}
\usepackage[USenglish]{babel}

\usepackage{amsmath}
\usepackage{amsthm}
\usepackage[capitalise, noabbrev, english]{cleveref}
\usepackage[overwrite]{crefthe}
\usepackage{csquotes}
\usepackage[
  type=CC,
  modifier=by,
  version=4.0,
  hyperxmp=false,
]{doclicense}
\usepackage{doi}
\usepackage{ebproof}
\usepackage{enumerate}
\usepackage{fnpct}
\usepackage{fontawesome5}
\usepackage{footnotebackref}
\usepackage[scr=pxtx, scrscaled=0.93113]{mathalpha}
\usepackage{mathpartir}
\usepackage{mathtools}
\usepackage[babel]{microtype}
\usepackage{subcaption}
\usepackage{tasks}
\usepackage{thmtools}
\usepackage{tikz}

\usepackage{mathstyle}

\DeclareTextCommandDefault{\textasciiasterisk}{*}
\DeclareMathSymbol{\to}{\mathbin}{symbols}{"21}
\DeclareMathSymbol{\mathcomma}{\mathpunct}{letters}{"3B}
\DeclareSymbolFontAlphabet\logic{bold}
\NewDocumentCommand\Sp{m}{\cramped{\sp{#1}}}
\NewDocumentCommand\Sb{m}{\cramped{\sb{#1}}}
\DeclareMathOperator\depth{depth}

\AddToHook{env/thebibliography/begin}
  {\phantomsection\addcontentsline{toc}{section}{\refname}}

\DeclareSymbolFont{AMSa}{U}{msa}{m}{n}
\DeclareSymbolFont{AMSb}{U}{msb}{m}{n}
\DeclareMathDelimiter{\lrcorner}{\mathclose}{AMSa}{"79}{AMSa}{"79}
\DeclareMathSymbol{\subsetneq}{\mathrel}{AMSb}{"28}

\numberwithin{equation}{section}
\AddToHook{begindocument}{\let\Label\label}

\crefname{page}{p.\@}{pp.\@}
\newcounter{rule}
\crefthename{rule}[the]{rule}[the]{rules}
\Crefthename{rule}[The]{rule}[The]{rules}
\creflabelformat{rule}{(#2#1#3)}

\newcounter{mark}

\pdfglyphtounicode{tfm:cmex10/parenleftbig}{0028}
\pdfglyphtounicode{tfm:cmex10/parenrightbig}{0029}
\pdfglyphtounicode{tfm:cmex10/bracketleftbig}{005B}
\pdfglyphtounicode{tfm:cmex10/bracketrightbig}{005D}
\pdfglyphtounicode{tfm:cmex10/braceleftbig}{007B}
\pdfglyphtounicode{tfm:cmex10/vextendsingle}{007C}
\pdfglyphtounicode{tfm:cmex10/bracerightbig}{007D}

\RenewDocumentCommand\mktextquote{m +m m m m m}{#1#2#4#5#3#6}

\AddToHook{begindocument}{
  \UseName{check@mathfonts}
  \ebproofset{
    center=true,
    template=$\cramped[\displaystyle]{\inserttext}$,
    right label template=\normalfont\inserttext,
    label separation=\fontdimen2\font,
    rule thickness=\fontdimen8\textfont3\relax,
  }
}
\ebproofnewstyle{strut}{
  template=$
    \xmathstrut{\fpeval{#1 - 1.0}}
    \smash{\cramped[\displaystyle]{\inserttext}}
  $,
}

\AddToHook{env/prooftree/begin}{\UseName{@restorepar}}

\MultVariant\footref

\pdfglyphtounicode{tfm:fa5free3solid/times}{2716}

\ExpandArgs{c}\let{Backref@HyperSymbol}\relax 

\DeclareFontFamily{U}{lasy}{}
\DeclareFontShape{U}{lasy}{m}{n}{<-> s * lasy10}{}
\DeclareSymbolFont{lasy}{U}{lasy}{m}{n}
\DeclareMathSymbol{\mdsmwhtsquare}{\mathord}{lasy}{50}
\DeclareMathSymbol{\mdwhtdiamond}{\mathord}{lasy}{51}
\pdfglyphtounicode{tfm:lasy10/a50}{25FB}
\pdfglyphtounicode{tfm:lasy10/a51}{2B26}
\newlength\lasyOffset
\NewDocumentCommand\UseLasy{m}{\mathpalette\UseLasyIn#1}
\NewDocumentCommand\UseLasyIn{m m}{\raisebox{\lasyOffset}{\hbox{$\UseName{m@th}#1#2$}}}

\pdfglyphtounicode{tfm:txsy/D}{1D49F}

\DeclarePairedDelimiter{\parens}{\lparen}{\rparen}
\DeclarePairedDelimiter{\braces}{\lbrace}{\rbrace}
\DeclarePairedDelimiter{\brackets}{\lbrack}{\rbrack}

\SetExpansion[context=bibliography, stretch=25, shrink=20]{encoding=*}{}
\AddToHook{env/thebibliography/begin}{\microtypecontext{expansion=bibliography}}

\DeclareCaptionLabelSeparator{labelsep}{\hskip0.9\labelsep}

\settasks{
  item-indent=\parindent,
  label-width=\parindent,
  label-align=right,
  label-offset=0pt,
}

\declaretheorem[numberwithin=section]{theorem}
\declaretheorem[sibling=theorem]{lemma}
\declaretheorem[style=definition, sibling=theorem]{definition}
\declaretheorem[style=remark, numbered=no, qed=$\lrcorner$]{example, remark}
\declaretheoremstyle[
  headfont=, 
  headformat=\textit{\NAME}\ExpandArgs{e}\IfBlankTF{\NOTE}{~\NUMBER}{\NOTE},
  notebraces=\empty\empty,
  postheadspace=\labelsep,
]{case}
\newcounter{proof}
\AddToHook{env/proof/begin}{\stepcounter{proof}}
\declaretheorem[style=case, numberwithin=proof]{case}

\usetikzlibrary{
  arrows.meta,
  decorations.pathmorphing,
  trees,
}
\tikzset{
  every picture/.append style={semithick, line cap=round},
  clap/.style={trim left=#1, trim right=#1},
  clap/.default=0pt,
  math/.style={execute at begin node=$, execute at end node=$},
  leadsto/.pic={
    \draw [
      -{To[width=2ex, length=0.3*2em]},
      decorate,
      decoration={snake, segment length=2em/2, amplitude=2ex/5},
      pic actions,
    ] (0, 0) to +(2em, 0);
  },
}

\NewDocumentCommand\titech{}
  {\institute{Tokyo Institute of Technology\\ Tokyo, Japan}}
\title{%
  Syntactic Cut-Elimination for
  Provability Logic\/~\texorpdfstring{$\logic{GL}$}{GL}\break\space
  via~Nested Sequents%
}
\author{%
  Akinori Maniwa\titech \email{maniwa.a.aa@m.titech.ac.jp}\and
  Ryo Kashima\titech \email{kashima@is.titech.ac.jp}%
}
\ExpandArgs{Nc}\NewCommandCopy\titlerunning{@title}
\NewDocumentCommand\authorrunning{}{A.~Maniwa \&~R.~Kashima}
\hypersetup{
  pdftitle={Syntactic Cut-Elimination for Provability Logic GL via~Nested Sequents},
  pdfauthor={%
    Akinori Maniwa <maniwa.a.aa@m.titech.ac.jp>,
    Ryo Kashima <kashima@is.titech.ac.jp>%
  },
  pdfkeywords={%
    Cut-elimination,
    Provability~logic,
    Nested~sequents,
    Proof~theory%
  },
  pdfsubject=\event,
  bookmarksnumbered=true,
}

\begin{document}

\maketitle

\begin{abstract}
  The cut-elimination procedure for the provability logic is known to be
  problematic: a Löb-like rule keeps cut-formulae intact on reduction,
  even in the principal case, thereby complicating the proof of termination.
  In this paper, we present a syntactic cut-elimination proof
  based on nested sequents, a generalization of sequents that
  allows a sequent to contain other sequents as single elements.
  A similar calculus was developed by Poggiolesi~(\UseName{hyper@link}{cite}{cite.poggiolesi2009purely}{2009}), but
  there are certain ambiguities in the proof.
  Adopting the idea of Kushida~(\UseName{hyper@link}{cite}{cite.kushida2020proof}{2020})
  into nested sequents,
  our proof does not require an extra measure on cuts or
  error-prone, intricate rewriting on derivations, but only
  straightforward inductions, thus leading to less ambiguity and confusion.
\end{abstract}

\begin{list}{}{%
  \def\and{\unskip\nobreakspace\textperiodcentered\space}%
  \rightmargin\leftmargin
}
  \item
    \nointerlineskip\leavevmode
\end{list}

\section{Introduction}

The provability logic~$\logic{GL}$, named after Gödel and Löb, is a modal logic
extending $\logic K$ with the Löb axiom $\UseLasy\mdsmwhtsquare (\UseLasy\mdsmwhtsquare A \to A) \to \UseLasy\mdsmwhtsquare A$, where
$\UseLasy\mdsmwhtsquare A$ can be roughly read as \textquote{$A$ is provable in Peano arithmetic}
(see, e.g.\@, Boolos~\cite{boolos1994logic} for more details).
Computationally, the Löb axiom represents a kind of
recursion~(e.g.\@,~\cite{kavvos2021intensionality,nakano2001fixed-point}),
and indeed in Kripke semantics the axiom is
interpreted as just an induction on its model.
Therefore, $\logic{GL}$ exhibits a certain \textquote{recursiveness} as its nature.

From a proof-theoretical viewpoint, a sequent calculus for $\logic{GL}$ is
obtained by the following single modal rule~\cite{leivant1981proof}:
\begin{prooftree*}[right label template=\rlap{\normalfont\inserttext}]
  \hypo { \Gamma, \UseLasy\mdsmwhtsquare \Gamma, \UseLasy\mdsmwhtsquare A \Rightarrow A }
  \infer1[(GLR)] { \Gamma', \UseLasy\mdsmwhtsquare \Gamma \Rightarrow \UseLasy\mdsmwhtsquare A, \Delta }
  \def\RuleName{GLR}
  \rewrite{%
    \raisebox\baselineskip{%
      \llap{%
        \ExpandArgs{c}\let{@currentlabelname}\RuleName%
        \refstepcounter{rule}\Label{rule:GLR}%
      }%
    }%
    \box\treebox
  }
\end{prooftree*}
where $\UseLasy\mdsmwhtsquare A$ is called \emph{diagonal formula}.
It is not difficult to prove the cut-elimination theorem using
semantical arguments~\cite{sambin+1982modal,avron1984modal}, but syntactically,
the diagonal formula is quite problematic: it appears
in both the premise and the conclusion of~\labelcref*{rule:GLR}, so
the standard double induction on the cut-formula and height fails.

Valentini~\cite{valentini1983modal} proposed a proof
using a third induction parameter, called the \emph{width} of a cut,
to justify the reduction involving \cref*{rule:GLR}\@.
Nevertheless, Valentini's proof is very brief and
only describes the principal cases for~\labelcref*{rule:GLR},
which raised a question about its termination
(see~\cite{goré+2012valentinis,moen2001proposed}).
In response, Goré and Ramanayake~\cite{goré+2012valentinis} confirmed
the validity of Valentini's arguments
by carefully analyzing the notion of width, but also pointed out,
overlooked by Valentini, that the width can be increased
by reduction in some cases~\cite[Remark~21]{goré+2012valentinis}.
Although such an increase is
certainly acceptable~\cite[Lemma~19]{goré+2012valentinis},
it makes their proof more complicated and non-trivial.

In this paper, we propose a more clarifying approach to
syntactic cut-elimination for $\logic{GL}$; unlike Valentini's,
our calculus is based on \emph{nested sequents}~\cite
  {kashima1994cut-free,brünnler2009deep,poggiolesi2009method}.
A similar calculus was developed by Poggiolesi~\cite{poggiolesi2009purely},
along with a syntactic cut-elimination proof
using a third induction parameter specific to its sequent structure.
This proof is rather simple and seemingly sound, but there are still
certain ambiguities around the third parameter (see \cref{sec:problem}),
and thus the termination is again imprecise.
These matters suggest that
while an additional measure on cuts could indeed resolve the problem,
it would not necessarily lead to a straightforward triple-induction proof, but
might require more careful checks, even where there seems less troublesome.

Instead of following the triple-induction approach,
we adopt the idea presented in Kushida~\cite{kushida2020proof},
also in Borga~\cite{borga1983some}, of introducing a subprocedure,
called \emph{diagonal-formula-elimination} in this paper, that removes
the diagonal formula in the premise prior to the reduction in question.
This helps us avoid the problematic cut-reduction and
recover the standard double induction proof of cut-elimination.
The advantage of nested sequents in employing this method is that,
thanks to their sequent structure,
it is easier to talk about where the diagonal formula is used in a derivation.
We take this advantage further by introducing \emph{annotations},
demonstrating that the nested-sequent basis allows
for much more concise and clear arguments,
based solely on a series of straightforward inductions.

The rest of the paper is organized as follows:
\Cref{sec:calculus} defines a nested sequent calculus for $\logic{GL}$, and
\cref{sec:problem} illustrates the problem on the cut-reduction method and
gives an overview of our approach.
\Cref{sec:annotation} introduces an auxiliary calculus with
additional information on the use of the diagonal formula.
\Cref{sec:diagonal-elim} demonstrates the procedure
for eliminating diagonal formulae, which
leads to the cut-elimination theorem in \cref{sec:cut-elimination}.
Finally, we conclude with some discussions in \cref{sec:conclusion}.

\section{The Calculus}\label{sec:calculus}

In this section,
we introduce a nested sequent calculus for $\logic{GL}$ with a one-sided formulation.
Our system is not very special as a nested sequent calculus,
so we only give a brief description here.
For a more detailed and general introduction to nested sequent calculus itself,
see, e.g.\@, Brünnler~\cite{brünnler2009deep}.

A \emph{formula} is defined by the following grammar:
\[
  A, B \Coloneq \left.
    \alpha
  \mathrel{} \mathclose{} \middle \vert \mathopen{} \mathrel{}
    \alpha \cramped{\sp\bot}
  \mathrel{} \mathclose{} \middle \vert \mathopen{} \mathrel{}
    A \land B
  \mathrel{} \mathclose{} \middle \vert \mathopen{} \mathrel{}
    A \lor B
  \mathrel{} \mathclose{} \middle \vert \mathopen{} \mathrel{}
    \UseLasy\mdsmwhtsquare A
  \mathrel{} \mathclose{} \middle \vert \mathopen{} \mathrel{}
    \UseLasy\mdwhtdiamond A
    \mkern1mu\text,
  \right.
\]
where $\alpha$ and $\alpha\cramped{\sp\bot}$ denote positive and negative atoms respectively,
both taken from a certain countable set.
The \emph{negation}~$ A\cramped{\sp\bot}$ of a formula~$A$ is defined inductively
in the usual way of extending duality on atoms using De Morgan's laws.
We may use $A \to B$ as an abbreviation for $A \cramped{\sp\bot}\lor B$.

A \emph{(nested) sequent} is defined by the following grammar:
\[
  \Gamma, \Delta \Coloneq \left.
    \mkern\medmuskip\mathord\cdot\vphantom{x}\mkern\medmuskip
  \mathrel{} \mathclose{} \middle \vert \mathopen{} \mathrel{}
    \Gamma, A
  \mathrel{} \mathclose{} \middle \vert \mathopen{} \mathrel{}
    \Gamma, \brackets{\Delta}
    \mkern1mu\text,
  \right.
\]
where \textquote{$\mkern\medmuskip\mathord\cdot\vphantom{x}\mkern\medmuskip$} denotes an empty sequent, and
the notation~$\brackets{\Delta}$ indicates that the sequent~$\Delta$ is
being placed as an element in another sequent (i.e.\@, \emph{nested}).
We may apply exchange implicitly as usual, so for example, we identify
$A \mathcomma\penalty\binoppenalty B \mathcomma\penalty\binoppenalty \brackets{C} \mathcomma\penalty\binoppenalty \brackets[\big]{D \mathcomma\penalty\binoppenalty \brackets{E}}$ with
$B \mathcomma\penalty\binoppenalty \brackets{C} \mathcomma\penalty\binoppenalty A \mathcomma\penalty\binoppenalty \brackets[\big]{\brackets{E} \mathcomma\penalty\binoppenalty D}$.
The juxtaposition of two sequents $\Gamma$ and $\Delta$ is written
simply with a comma~\textquote{$\mathcomma$} as $\Gamma, \Delta$ as usual.

Intuitively, a nested sequent represents a tree consisting of ordinary sequents
(i.e.\@, multisets of formulae) by means of the bracket~$\brackets{\mathord{-}}$ nesting;
for example, the sequent
$A \mathcomma\penalty\binoppenalty \brackets[\big]{\brackets{B} \mathcomma\penalty\binoppenalty \brackets{C \mathcomma\penalty\binoppenalty D}} \mathcomma\penalty\binoppenalty \brackets[\big]{E \mathcomma\penalty\binoppenalty \brackets{F \mathcomma\penalty\binoppenalty G \mathcomma\penalty\binoppenalty H}}$
corresponds to the following tree structure:
\begin{gather*}
  \tikz[
    clap, grow cyclic,
    level distance=6.5ex,
    level/.append style={
      sibling angle={120 / #1},
      clockwise from={90 + 60 / #1},
    },
    every node/.append style={math},
    edge from parent/.append style={-Stealth},
    inner sep=0pt, outer sep=0.4em,
  ] \node { A }
    child {
      node[circle] { \mkern\medmuskip\mathord\cdot\vphantom{x}\mkern\medmuskip }
        child { node { B } }
        child { node { C, D } }
    }
    child {
      node[circle] { E }
        child[missing]
        child { node { F, G, H \enskip } }
    };
\end{gather*}
By considering such a tree as a shape within a Kripke model,
we can obtain modal rules directly from Kripke semantics.
From the perspective of structural proof theory, on the other hand,
$\brackets{\mathord{-}}$ is a structure corresponding to~$\UseLasy\mdsmwhtsquare$,
just as \textquote{$\mathcomma$} is to~$\lor$,
allowing for modal reasoning in an analytic way.

Before getting into our proof system,
we need to introduce the notion of \emph{context}.
A \emph{unary context} is informally
a sequent with a single \emph{hole}~$\braces{\mathord{-}}$ as a placeholder,
formally defined by the following grammar:
\[
  \Gamma\braces*{\mathord{-}} \Coloneq \left. \Delta, \braces*{\mathord{-}}
  \mathrel{} \mathclose{} \middle \vert \mathopen{} \mathrel{}
  \Delta, \brackets[\big]{\Gamma\braces*{\mathord{-}}}
  \mkern1mu\text. \right.
\]
Given a unary context~$\Gamma\braces{\mathord{-}}$ and a sequent~$\Delta$,
we write $\Gamma\braces{\Delta}$ for the sequent obtained
by replacing $\braces{\mathord{-}}$ with $\Delta$ in $\Gamma\braces{\mathord{-}}$.
For instance, $\Gamma\braces{\mathord{-}} \equiv A \mathcomma\penalty\binoppenalty \brackets{B \mathcomma\penalty\binoppenalty C} \mathcomma\penalty\binoppenalty \brackets[\big]{D \mathcomma\penalty\binoppenalty \braces{\mathord{-}}}$
is a unary context, and then $\Gamma\braces[\big]{E, \brackets{F, G}}$ represents
the sequent $A \mathcomma\penalty\binoppenalty \brackets{B \mathcomma\penalty\binoppenalty C} \mathcomma\penalty\binoppenalty \brackets[\big]{D \mathcomma\penalty\binoppenalty E \mathcomma\penalty\binoppenalty \brackets{F \mathcomma\penalty\binoppenalty G}}$.
When filling an empty sequent into a context,
we omit its symbol~\textquote{$\mkern\medmuskip\mathord\cdot\vphantom{x}\mkern\medmuskip$} from the result;
that is, $\Gamma\braces{}$ means the sequent~$\Gamma\braces{\mkern\medmuskip\mathord\cdot\vphantom{x}\mkern\medmuskip}$, which
is of course also distinguished from the context~$\Gamma\braces{\mathord{-}}$.
A \emph{binary context} $\Gamma\braces{\cramped{\mathord{-}\sb{1}}}\braces{\cramped{\mathord{-}\sb{2}}}$,
a sequent with two distinct holes of~$\braces{\cramped{\mathord{-}\sb{1}}}$ and~$\braces{\cramped{\mathord{-}\sb{2}}}$,
is formally defined and used in a similar way.

\begin{definition}[Depth]
  The \emph{depth} of a unary context is defined inductively as follows:
  \begin{alignat*}{2}
    &\depth\parens[\big]{\Delta, \braces*{\mathord{-}}} &&= 0 \mkern1mu\text; \\
    &\depth\parens[\big]{\Delta, \brackets*{\Gamma\braces*{\mathord{-}}}} &&=
      \depth\parens[\big]{\Gamma\braces*{\mathord{-}}} + 1 \mkern1mu\text.
  \end{alignat*}
  It is, in short,
  the nesting depth of the bracket~$\brackets{\mathord{-}}$ at the hole~$\braces{\mathord{-}}$ position.
\end{definition}

\Cref{fig:rules} shows the inference rules of our system.
\begin{figure}[t]
  \begin{mathpar}
    \begin{prooftree}
      \hypo \mathstrut
      \infer1[(id)] { \Gamma\braces*{\alpha\cramped{\sp\bot}, \alpha} }
      \def\RuleName{id}
      \rewrite{%
        \raisebox\baselineskip{%
          \llap{%
            \ExpandArgs{c}\let{@currentlabelname}\RuleName%
            \refstepcounter{rule}\Label{rule:id}%
          }%
        }%
        \box\treebox
      }
    \end{prooftree}
    \and
    \begin{prooftree}
      \hypo    { \Gamma\braces*{A} }
      \hypo    { \Gamma\braces*{B} }
      \infer2[($\land$)] { \Gamma\braces*{A \land B} }
      \def\RuleName{$\land$}
      \rewrite{%
        \raisebox\baselineskip{%
          \llap{%
            \ExpandArgs{c}\let{@currentlabelname}\RuleName%
            \refstepcounter{rule}\Label{rule:and}%
          }%
        }%
        \box\treebox
      }
    \end{prooftree}
    \and
    \begin{prooftree}
      \hypo   { \Gamma\braces*{A, B} }
      \infer1[($\lor$)] { \Gamma\braces*{A \lor B} }
      \def\RuleName{$\lor$}
      \rewrite{%
        \raisebox\baselineskip{%
          \llap{%
            \ExpandArgs{c}\let{@currentlabelname}\RuleName%
            \refstepcounter{rule}\Label{rule:or}%
          }%
        }%
        \box\treebox
      }
    \end{prooftree}
    \\
    \begin{prooftree}
      \hypo    { \Gamma\braces[\big]{\left\lbrack \UseLasy\mdwhtdiamond  A\cramped{\sp\bot}, A \right\rbrack} }
      \infer1[($\UseLasy\mdsmwhtsquare$)] { \Gamma\braces*{\UseLasy\mdsmwhtsquare A} }
      \def\RuleName{$\UseLasy\mdsmwhtsquare$}
      \rewrite{%
        \raisebox\baselineskip{%
          \llap{%
            \ExpandArgs{c}\let{@currentlabelname}\RuleName%
            \refstepcounter{rule}\Label{rule:box}%
          }%
        }%
        \box\treebox
      }
    \end{prooftree}
    \and
    \begin{prooftree}
      \hypo    { \Gamma\braces[\big]{\Delta\braces*{A}, \UseLasy\mdwhtdiamond A} }
      \infer1[($\UseLasy\mdwhtdiamond$)\enskip if~$\depth(\Delta\braces*{\mathord{-}}) > 0$]
        { \Gamma\braces[\big]{\Delta\braces*{}, \UseLasy\mdwhtdiamond A} }
      \def\RuleName{$\UseLasy\mdwhtdiamond$}
      \rewrite{%
        \raisebox\baselineskip{%
          \llap{%
            \ExpandArgs{c}\let{@currentlabelname}\RuleName%
            \refstepcounter{rule}\Label{rule:dia}%
          }%
        }%
        \box\treebox
      }
    \end{prooftree}
  \end{mathpar}
  \par\removelastskip
  \caption{Inference rules for $\logic{GL}$.}\label{fig:rules}
\end{figure}
The non-modal rule are fairly standard, except for the form of sequents.
\Cref{rule:box} is a kind of the Löb rule, as is~\cref{rule:GLR}, and
is the only rule in the system that consumes a~$\brackets{\mathord{-}}$.
Reflecting the transitivity of $\logic{GL}$-models, we can deduce $\UseLasy\mdwhtdiamond A$
by \cref{rule:dia} from $A$ at a deeper location within several $\brackets{\mathord{-}}$'s.
A contraction for~$\UseLasy\mdwhtdiamond A$ is incorporated into \cref*{rule:dia}
to ensure the admissibility of contraction~(\cref{claim:admissible}).

This is indeed a complete proof system for $\logic{GL}$ in the following sense,
but we omit the proof here.

\begin{theorem}[Completeness]
  A formula $A$ is a theorem of\/ $\logic{GL}$
  if and only if the sequent $A$ is provable in the calculus.
\end{theorem}

\begin{example}[The \NoCaseChange{Löb} axiom]\label{example:löb}
  A proof of
  \begin{math}
    \UseLasy\mdsmwhtsquare (\UseLasy\mdsmwhtsquare \alpha \to \alpha) \to \UseLasy\mdsmwhtsquare \alpha \equiv
    \UseLasy\mdwhtdiamond \parens{\UseLasy\mdsmwhtsquare \alpha \land \alpha\cramped{\sp\bot}} \lor \UseLasy\mdsmwhtsquare \alpha
  \end{math}
  is as follows:
  \begin{prooftree*}[strut=1.03]
    \infer0[(id)]  { \UseLasy\mdwhtdiamond \parens{\UseLasy\mdsmwhtsquare \alpha \land \alpha\cramped{\sp\bot}}, \brackets[\big]{\UseLasy\mdwhtdiamond \alpha\cramped{\sp\bot}, \alpha, \brackets*{\UseLasy\mdwhtdiamond \alpha\cramped{\sp\bot}, \alpha, \alpha\cramped{\sp\bot}}} }
    \infer1[($\UseLasy\mdwhtdiamond$)] { \UseLasy\mdwhtdiamond \parens{\UseLasy\mdsmwhtsquare \alpha \land \alpha\cramped{\sp\bot}}, \brackets[\big]{\UseLasy\mdwhtdiamond \alpha\cramped{\sp\bot}, \alpha, \brackets*{\UseLasy\mdwhtdiamond \alpha\cramped{\sp\bot}, \alpha}} }
    \infer1[($\UseLasy\mdsmwhtsquare$)] { \UseLasy\mdwhtdiamond \parens{\UseLasy\mdsmwhtsquare \alpha \land \alpha\cramped{\sp\bot}}, \brackets*{\UseLasy\mdwhtdiamond \alpha\cramped{\sp\bot}, \alpha, \UseLasy\mdsmwhtsquare \alpha} }
    \set{right label template=\rlap{\normalfont\inserttext}}
    \infer0[(id)]  { \UseLasy\mdwhtdiamond \parens{\UseLasy\mdsmwhtsquare \alpha \land \alpha\cramped{\sp\bot}}, \brackets*{\UseLasy\mdwhtdiamond \alpha\cramped{\sp\bot}, \alpha, \alpha\cramped{\sp\bot}} }
    \infer2[($\land$)] { \UseLasy\mdwhtdiamond \parens{\UseLasy\mdsmwhtsquare \alpha \land \alpha\cramped{\sp\bot}}, \brackets*{\UseLasy\mdwhtdiamond \alpha\cramped{\sp\bot}, \alpha, \UseLasy\mdsmwhtsquare \alpha \land \alpha\cramped{\sp\bot}} }
    \infer1[($\UseLasy\mdwhtdiamond$)] { \UseLasy\mdwhtdiamond \parens{\UseLasy\mdsmwhtsquare \alpha \land \alpha\cramped{\sp\bot}}, \brackets*{\UseLasy\mdwhtdiamond \alpha\cramped{\sp\bot}, \alpha} }
    \infer1[($\UseLasy\mdsmwhtsquare$)] { \UseLasy\mdwhtdiamond \parens{\UseLasy\mdsmwhtsquare \alpha \land \alpha\cramped{\sp\bot}}, \UseLasy\mdsmwhtsquare \alpha }
    \infer1[($\lor$)]  { \UseLasy\mdwhtdiamond \parens{\UseLasy\mdsmwhtsquare \alpha \land \alpha\cramped{\sp\bot}} \lor \UseLasy\mdsmwhtsquare \alpha }
    \AddToHookNext{env/equation*/end}{\qedhere}
  \end{prooftree*}
\end{example}

\begin{remark}
  Poggiolesi~\cite{poggiolesi2009purely} also developed
  a nested sequent calculus for $\logic{GL}$, but
  under the name~\emph{tree-hypersequents}, with a two-sided representation.
  The main difference\footnote{
    Another difference is a form of cut, which
    shall be discussed in \cref{sec:conclusion}.
  } is the rules for~$\UseLasy\mdwhtdiamond$ (or, the left rules for~$\UseLasy\mdsmwhtsquare$).
  Poggiolesi instead employed
  the following two rules (but in our notation):
  \begin{mathpar}
    \begin{prooftree}
      \hypo   { \Gamma\braces[\big]{\brackets*{\Delta, \UseLasy\mdsmwhtsquare A\cramped{\sp\bot}}, \UseLasy\mdwhtdiamond A} }
      \hypo   { \Gamma\braces[\big]{\brackets*{\Delta, A}, \UseLasy\mdwhtdiamond A} }
      \infer2 { \Gamma\braces[\big]{\brackets*{\Delta}, \UseLasy\mdwhtdiamond A} }
    \end{prooftree}
    \and
    \begin{prooftree}
      \hypo   { \Gamma\braces[\big]{\brackets*{\Delta, \UseLasy\mdwhtdiamond A}, \UseLasy\mdwhtdiamond A} }
      \infer1 { \Gamma\braces[\big]{\brackets*{\Delta}, \UseLasy\mdwhtdiamond A} }
    \end{prooftree}
  \end{mathpar}
  Nevertheless, there is no essential difference, especially as for provability.
  We shall discuss Poggiolesi's cut-elimination proof
  in \cref{sec:problem}.
\end{remark}

\begin{definition}[Cut]\label{def:cut}
  A \emph{cut} in our calculus has the following form:
  \begin{prooftree*}[right label template=\rlap{\normalfont\inserttext}]
    \hypo    { \Gamma\braces*{A} }
    \hypo    { \Gamma\braces*{A\cramped{\sp\bot}} }
    \infer2[(cut)] { \Gamma\braces*{} }
    \def\RuleName{cut}
    \rewrite{%
      \raisebox\baselineskip{%
        \llap{%
          \ExpandArgs{c}\let{@currentlabelname}\RuleName%
          \refstepcounter{rule}\Label{rule:cut}%
        }%
      }%
      \box\treebox
    }
  \end{prooftree*}
\end{definition}

The \emph{height} of a derivation is defined in the standard manner, i.e.\@,
the maximum length of consecutive applications of inference rules
in that derivation.
Several common rules required for the cut-elimination procedure
are shown to be (height-preserving) admissible:

\begin{lemma}[Inversion]\label{claim:inversion}
  The following rules are height-preserving admissible:
  \begin{mathpar}
    \begin{prooftree}
      \hypo{ \Gamma\braces*{A \land B} }
      \infer1[($\land$)\textsuperscript{\textminus1}] { \Gamma\braces*{A} }
    \end{prooftree}
    \and
    \begin{prooftree}
      \hypo{ \Gamma\braces*{A \land B} }
      \infer1[($\land$)\textsuperscript{\textminus1}] { \Gamma\braces*{B} }
    \end{prooftree}
    \and
    \begin{prooftree}
      \hypo     { \Gamma\braces*{A \lor B} }
      \infer1[($\lor$)\textsuperscript{\textminus1}] { \Gamma\braces*{A, B} }
    \end{prooftree}
    \and
    \begin{prooftree}
      \hypo      { \Gamma\braces*{\UseLasy\mdsmwhtsquare A} }
      \infer1[($\UseLasy\mdsmwhtsquare$)\textsuperscript{\textminus1}] { \Gamma\braces[\big]{\brackets*{\UseLasy\mdwhtdiamond A\cramped{\sp\bot}, A}} }
    \end{prooftree}
  \end{mathpar}
\end{lemma}

\begin{proof}
  By induction on derivation.
\end{proof}

\begin{lemma}[Identity]\label{claim:identity}
  The following rule is admissible:
  \begin{prooftree*}[right label template=\rlap{\normalfont\inserttext}]
    \infer0[(id)] { \Gamma\braces*{A\cramped{\sp\bot}, A} }
  \end{prooftree*}
\end{lemma}

\begin{proof}
  By induction on $A$.
\end{proof}

\begin{lemma}[Structural rules]\label{claim:admissible}
  The following rules are height-preserving admissible:
  \begin{mathpar}
    \begin{prooftree}
      \hypo     { \Gamma\braces*{} }
      \infer1[(weak)] { \Gamma\braces*{\Delta} }
    \end{prooftree}
    \and
    \begin{prooftree}
      \hypo         { \Gamma\braces*{A, A} }
      \infer1[(contract)] { \Gamma\braces*{A} }
    \end{prooftree}
    \and
    \begin{prooftree}
      \hypo       { \Gamma\braces[\big]{\Delta\braces*{}, \brackets*{\Delta'}} }
      \infer1[(rebase)\enskip if~$\depth(\Delta\braces*{\mathord{-}}) > 0$]
        { \Gamma\braces[\big]{\Delta\braces*{\Delta'}} }
      \def\RuleName{rebase}
      \rewrite{%
        \raisebox\baselineskip{%
          \llap{%
            \ExpandArgs{c}\let{@currentlabelname}\RuleName%
            \refstepcounter{rule}\Label{rule:rebase}%
          }%
        }%
        \box\treebox
      }
    \end{prooftree}
  \end{mathpar}
\end{lemma}

\begin{proof}
  By induction on derivation, along with \cref{claim:inversion}.
\end{proof}

Semantically, the rebasing rule is to instantiate an arbitrary transition
denoted by $\brackets{\mathord{-}}$ into a more concrete one described by $\Delta\braces{\mathord{-}}$, and
its side-condition corresponds to the transitivity, as in \cref{rule:dia}.

\section{Problem on Cut-Reduction}\label{sec:problem}

In this section, we explain why the standard cut-reduction method
does not work as expected for $\logic{GL}$, even in nested sequents,
together with the problem with Poggiolesi's proof.
We also give an overview of our approach
with an informal description of our rewriting procedure.

The standard double induction fails in the principal case of $\UseLasy\mdwhtdiamond$ and $\UseLasy\mdsmwhtsquare$:
\begin{gather}
  \begin{prooftree}[center=false]
    \small
    \begingroup
      \set{rule style=no rule, template=}
      \begingroup
        \set{rule margin=0pt}
        \hypo{}
        \ellipsis{\small$\mathscr{D}_1$}{}
      \endgroup
    \endgroup
    \infer[no rule]1 { \Gamma\braces*{\Delta\braces*{A\cramped{\sp\bot}}, \UseLasy\mdwhtdiamond A\cramped{\sp\bot}} }
    \infer1[($\UseLasy\mdwhtdiamond$)] { \Gamma\braces*{\Delta\braces*{}, \UseLasy\mdwhtdiamond A\cramped{\sp\bot}} }
    \set{right label template=\rlap{\normalfont\inserttext}}
    \begingroup
      \set{rule style=no rule, template=}
      \begingroup
        \set{rule margin=0pt}
        \hypo{}
        \ellipsis{\small$\mathscr{D}_2$}{}
      \endgroup
    \endgroup
    \infer[no rule]1{ \Gamma\braces*{\Delta\braces*{}, \left\lbrack \UseLasy\mdwhtdiamond A\cramped{\sp\bot}, A \right\rbrack} }
    \infer1[($\UseLasy\mdsmwhtsquare$)] { \Gamma\braces*{\Delta\braces*{}, \UseLasy\mdsmwhtsquare A} }
    \infer2[(cut)] { \Gamma\braces*{\Delta\braces*{}} }
    \def\RuleName{cut}
    \rewrite{%
      \raisebox\baselineskip{%
        \llap{%
          \ExpandArgs{c}\let{@currentlabelname}\RuleName%
          \refstepcounter{rule}\Label{tree:beta}%
        }%
      }%
      \box\treebox
    }
  \end{prooftree}
  \mathrlap{\qquad\tikz[baseline={(0,-\baselineskip)}] \pic[anchor=center]{leadsto};\quad}
  \nonumber
  \displaybreak[0] \\
  \begin{prooftree}[separation=0.5em, strut=1.05]
    \small
    \begingroup
      \set{rule style=no rule, template=}
      \begingroup
        \set{rule margin=0pt}
        \hypo{}
        \ellipsis{\small$\mathscr{D}_1$}{}
      \endgroup
    \endgroup
    \infer[no rule]1{ \Gamma\braces*{\Delta\braces*{A\cramped{\sp\bot}}, \UseLasy\mdwhtdiamond A\cramped{\sp\bot}} }
    \begingroup
      \set{rule style=no rule, template=}
      \begingroup
        \set{rule margin=0pt}
        \hypo{}
        \ellipsis{\small$\mathscr{D}_2$}{}
      \endgroup
    \endgroup
    \infer[no rule]1{ \Gamma\braces*{\Delta\braces*{}, \left\lbrack \UseLasy\mdwhtdiamond A\cramped{\sp\bot}, A \right\rbrack} }
    \infer1[($\UseLasy\mdsmwhtsquare$)] { \Gamma\braces*{\Delta\braces*{}, \UseLasy\mdsmwhtsquare A} }
    \infer[dashed]1[(weak)] { \Gamma\braces*{\Delta\braces*{A\cramped{\sp\bot}}, \UseLasy\mdsmwhtsquare A} }
    \infer2[(cut)\refstepcounter{mark}\ExpandArgs{c}\textsuperscript{@currentlabel}\Label{cut:1st}]  { \Gamma\braces*{\Delta\braces*{A\cramped{\sp\bot}}} }
    \infer[dashed]1[(weak)]   { \Gamma\braces*{\Delta\braces*{\UseLasy\mdwhtdiamond A\cramped{\sp\bot}, A\cramped{\sp\bot}}} }
    \begingroup
      \set{rule style=no rule, template=}
      \begingroup
        \set{rule margin=0pt}
        \hypo{}
        \ellipsis{\small$\mathscr{D}_2$}{}
      \endgroup
    \endgroup
    \infer[no rule]1{ \Gamma\braces*{\Delta\braces*{}, \left\lbrack \UseLasy\mdwhtdiamond A\cramped{\sp\bot}, A \right\rbrack} }
    \infer[dashed]1[(rebase)] { \Gamma\braces*{\Delta\braces*{\UseLasy\mdwhtdiamond A\cramped{\sp\bot}, A}} }
    \infer2[(cut)\refstepcounter{mark}\ExpandArgs{c}\textsuperscript{@currentlabel}\Label{cut:2nd}]    { \Gamma\braces*{\Delta\braces*{\UseLasy\mdwhtdiamond A\cramped{\sp\bot}}} }
    \begingroup
      \set{rule style=no rule, template=}
      \begingroup
        \set{rule margin=0pt}
        \hypo{}
        \ellipsis{\small$\mathscr{D}_2$}{}
      \endgroup
    \endgroup
    \infer[no rule]1{ \Gamma\braces*{\Delta\braces*{}, \left\lbrack \UseLasy\mdwhtdiamond A\cramped{\sp\bot}, A \right\rbrack} }
    \infer[dashed]1[(weak)]   { \Gamma\braces*{\Delta\braces*{\lbrack\rbrack}, \left\lbrack \UseLasy\mdwhtdiamond A\cramped{\sp\bot}, A \right\rbrack} }
    \infer[dashed]1[(rebase)] { \Gamma\braces*{\Delta\braces*{\left\lbrack \UseLasy\mdwhtdiamond A\cramped{\sp\bot}, A \right\rbrack}} }
    \infer1[($\UseLasy\mdsmwhtsquare$)] { \Gamma\braces*{\Delta\braces*{\UseLasy\mdsmwhtsquare A}} }
    \set{separation=1em}
    \infer2[(cut)\refstepcounter{mark}\ExpandArgs{c}\textsuperscript{@currentlabel}\Label{cut:3rd}]    { \Gamma\braces*{\Delta\braces*{}} }
  \end{prooftree}
  \label{reduction:beta}
\end{gather}
The first cut\footref{cut:1st} is admissible
because of the smaller derivation of the left premise, and
so is the second cut\footref{cut:2nd}
because of the smaller size of the cut-formula, but
neither is small for the third cut\footref{cut:3rd}.

\paragraph*{Naïve Attempt.}

Although the cut-formula stays the same,
it can be seen that on the third cut, compared to the original,
the cut-formula~$\UseLasy\mdwhtdiamond A\cramped{\sp\bot}$ has moved by the depth of $\Delta\braces{\mathord{-}}$ toward the leaves
of the tree represented by the sequent $\Gamma\braces{\Delta\braces{}}$.
So one might think that the reduction could be justified
by appealing to the remaining distance to the leaves, namely,
by induction on the following lexicographic ordering:
\begin{enumerate}[(i)]
  \item
    The size of the cut-formula;
  \item\label{cond:step}
    The maximum number of steps required for the cut-formula to reach a leaf;
    and
  \item\label{cond:total-height}
    The total height of the premise derivations.
\end{enumerate}
This approach, unfortunately, does not work as expected:
the well-foundedness of the ordering requires that a tree not grow its branches
due to the reductions admitted by condition~\eqref{cond:total-height},
which is in fact not true in the following case:
\begin{gather}
  \SwapAboveDisplaySkip
  \ebproofset{center=false}
  \begin{gathered}
    \begin{prooftree}
      \small
      \begingroup
        \set{rule style=no rule, template=}
        \begingroup
          \set{rule margin=0pt}
          \hypo{}
          \ellipsis{\small$\mathscr{D}_1$}{}
        \endgroup
      \endgroup
      \infer[no rule]1{ \Gamma\braces*{\left\lbrack \UseLasy\mdwhtdiamond A\cramped{\sp\bot}, A \right\rbrack}\braces*{B} }
      \infer1[($\UseLasy\mdsmwhtsquare$)] { \Gamma\braces*{\UseLasy\mdsmwhtsquare A}\braces*{B} }
      \def\RuleName{$\UseLasy\mdsmwhtsquare$}
      \rewrite{%
        \raisebox\baselineskip{%
          \llap{%
            \ExpandArgs{c}\let{@currentlabelname}\RuleName%
            \refstepcounter{rule}\Label{tree:com:left}%
          }%
        }%
        \box\treebox
      }
      \begingroup
        \set{rule style=no rule, template=}
        \begingroup
          \set{rule margin=0pt}
          \hypo{}
          \ellipsis{\small$\mathscr{D}_2$}{}
        \endgroup
      \endgroup
      \infer[no rule]1{ \Gamma\braces*{\UseLasy\mdsmwhtsquare A}\braces*{B\cramped{\sp\bot}} }
      \infer2[(cut)] { \Gamma\braces*{\UseLasy\mdsmwhtsquare A}\braces*{} }
      \def\RuleName{cut}
      \rewrite{%
        \raisebox\baselineskip{%
          \llap{%
            \ExpandArgs{c}\let{@currentlabelname}\RuleName%
            \refstepcounter{rule}\Label{tree:com}%
          }%
        }%
        \box\treebox
      }
    \end{prooftree}
    \nolinebreak\quad
    \tikz[baseline={(0,-\baselineskip)}] \pic[anchor=center]{leadsto};%
    \nolinebreak\hskip0.6667em\relax
    \begin{prooftree}
      \small
      \begingroup
        \set{rule style=no rule, template=}
        \begingroup
          \set{rule margin=0pt}
          \hypo{}
          \ellipsis{\small$\mathscr{D}_1$}{}
        \endgroup
      \endgroup
      \infer[no rule]1{ \Gamma\braces*{\left\lbrack \UseLasy\mdwhtdiamond A\cramped{\sp\bot}, A \right\rbrack}\braces*{B} }
      \begingroup
        \set{rule style=no rule, template=}
        \begingroup
          \set{rule margin=0pt}
          \hypo{}
          \ellipsis{\small$\mathscr{D}_2$}{}
        \endgroup
      \endgroup
      \infer[no rule]1{ \Gamma\braces*{\UseLasy\mdsmwhtsquare A}\braces*{B\cramped{\sp\bot}} }
      \infer[dashed]1[($\UseLasy\mdsmwhtsquare$)\textsuperscript{\textminus1}]
        { \Gamma\braces*{\left\lbrack \UseLasy\mdwhtdiamond A\cramped{\sp\bot}, A \right\rbrack}\braces*{B\cramped{\sp\bot}} }
      \infer2[(cut)]   { \Gamma\braces*{\left\lbrack \UseLasy\mdwhtdiamond A\cramped{\sp\bot}, A \right\rbrack}\braces*{} }
      \infer1[($\UseLasy\mdsmwhtsquare$)]   { \Gamma\braces*{\UseLasy\mdsmwhtsquare A}\braces*{} }
    \end{prooftree}
  \end{gathered}
  \label{reduction:com}
\end{gather}
This permutation exposes the~$\brackets{\mathord{-}}$ previously
discharged by~\labelcref{tree:com:left} in the left premise,
potentially increasing the measure~\eqref{cond:step}\footnote{
  For instance, take $\Gamma\braces{\cramped{\mathord{-}\sb{1}}}\braces{\cramped{\mathord{-}\sb{2}}} \equiv \braces{\cramped{\mathord{-}\sb{1}}}, \braces{\cramped{\mathord{-}\sb{2}}}$. Then,
  on the left-hand side the cut-formula~$B$ is already on the leaf, but
  on the right-hand side it can go one step ahead.
}.
Even were we to address this case by considering the number of~$\UseLasy\mdsmwhtsquare$'s as well
as of~$\brackets{\mathord{-}}$'s in~\eqref{cond:step}, it would impose a strong restriction
on weakening, thus breaking the argument in other cases.

\paragraph*{Poggiolesi's Approach.}

Poggiolesi~\cite{poggiolesi2009purely} proposed a similar triple-induction proof using
the notion of \emph{position} instead of the measure~\eqref{cond:step}.
The position is, in brief, a variant of~\eqref{cond:step} that
estimates the maximum number of steps by considering not just the end-sequent
but also all sequents appearing in a derivation,
whereby the reduction~\eqref{reduction:com} is no longer a problem.
However, this rather causes trouble with the third cut mentioned above.
More specifically, Poggiolesi's admissible
rule~(\smash{\~4})~\cite[Lemma~4.10]{poggiolesi2009purely},
analogous to our \cref[noun]{rule:rebase},
is used to move a subtree upwards on reduction
\cite[Lemma~4.26-Case~3.2-4(a)]{poggiolesi2009purely},
which can cause an increase in position
since rewriting subderivations affects the position on the end-sequent.
Such an operation is essential for
the interaction between $\UseLasy\mdwhtdiamond$ and $\UseLasy\mdsmwhtsquare$, and
the transitive property makes the trouble unavoidable.

Poggiolesi's approach
basically follows the work by Negri~\cite{negri2005proof},
presented a proof based on \emph{labeled sequents}
with an additional parameter called \emph{range},
similar to the position but defined in terms of labels.
Both position and range attempt to capture the well-foundedness of $\logic{GL}$-models
by means of their sequent structures, but there are crucial differences.
Negri used
\emph{label substitution}~\cite[Lemma~4.3]{negri2005proof}
to achieve the required transformation without increasing range, which
makes the triple-induction proof effective.
Here, it takes advantage of the fact that
a substitution yields a \emph{graph} structure rather than a tree,
and precisely for this reason,
such an operation cannot be fully reproduced in nested sequents.
Poggiolesi seems to have overlooked this point,
and consequently, without filling this gap, the argument
would be inadequate to simulate Negri's method.

\paragraph*{Our Approach.}

As in the case of Valentini,
a more detailed analysis might make up for this piece,
but contrary to Poggiolesi's expectation,
even in nested sequents it is not so obvious how to make triple induction work.
In addition, adding a third induction parameter is annoying
since it has a relatively broad impact on the overall induction, which
induces oversights in some boring cases such as permutations.
The triple-induction approach is not so ideal for these reasons, and
a more reliable method
based on more intuitive and purely syntactic concept is desirable.

The reason why the problematic third cut is necessary
is to eliminate the diagonal formula~$\UseLasy\mdwhtdiamond A\cramped{\sp\bot}$ in the premise
of~\labelcref*{rule:box};
if we could do that in any other way, then the problem should be resolved.
Kushida~\cite{kushida2020proof} showed that it is indeed possible
by relying on cuts only on~$A$ (notice, not on~$\UseLasy\mdwhtdiamond A\cramped{\sp\bot}$),
motivated by a syntactic cut-elimination proof for
provability logic~$\logic S$~\cite{solovay1976provability} in ordinary sequents.

Let us review the basic idea of Kushida~\cite{kushida2020proof},
but in the form of nested sequents.
Suppose $\Gamma\braces[\big]{\brackets{\UseLasy\mdwhtdiamond A\cramped{\sp\bot}, A}}$ is cut-free provable with a derivation~$\mathscr{D}$,
and consider dropping $\UseLasy\mdwhtdiamond A\cramped{\sp\bot}$ to obtain $\Gamma\braces{\lbrack A \rbrack}$.
If $\UseLasy\mdwhtdiamond A\cramped{\sp\bot}$ is not used in~$\mathscr{D}$ at all, then we can remove it
from all initial sequents of~$\mathscr{D}$ to obtain a derivation of~$\Gamma\braces{\lbrack A \rbrack}$.
Otherwise, there must be a pair of
relevant applications of~\labelcref{tree:orig:dia,tree:orig:box} in~$\mathscr{D}$,
as shown in \cref{fig:orig}.
\begin{figure}[t]
  \ebproofset{right label template=\rlap{\normalfont\inserttext}}
  \begin{subcaptionblock}[t]{0.45\textwidth}\centering
    \begin{prooftree}
      \begingroup
        \set{rule style=no rule, template=}
        \begingroup
          \set{rule margin=0pt}
          \hypo{}
          \ellipsis{\small$\mathscr{D}_3$}{}
        \endgroup
      \endgroup
      \infer[no rule]1{ \Gamma''\braces*{ \brackets*{\UseLasy\mdwhtdiamond A\cramped{\sp\bot}, \Delta'\braces[\big]{\brackets*{\UseLasy\mdwhtdiamond B\cramped{\sp\bot}, \Theta, A\cramped{\sp\bot}}}} } }
      \infer1[($\UseLasy\mdwhtdiamond$)\refstepcounter{mark}\ExpandArgs{c}\textsuperscript{@currentlabel}\Label{mark:orig:dia}]
        { \Gamma''\braces*{\brackets*{\UseLasy\mdwhtdiamond A\cramped{\sp\bot}, \Delta'\braces[\big]{\brackets*{\UseLasy\mdwhtdiamond B\cramped{\sp\bot}, \Theta}}}} }
        \def\RuleName{$\UseLasy\mdwhtdiamond$}
        \rewrite{%
          \raisebox\baselineskip{%
            \llap{%
              \ExpandArgs{c}\let{@currentlabelname}\RuleName%
              \refstepcounter{rule}\Label{tree:orig:dia}%
            }%
          }%
          \box\treebox
        }
      \ellipsis{\small$\mathscr{D}_2$} { \Gamma'\braces*{\brackets[\big]{\UseLasy\mdwhtdiamond A\cramped{\sp\bot}, \Delta\braces*{\brackets*{\UseLasy\mdwhtdiamond B\cramped{\sp\bot}, B}}}} }
      \infer1[($\UseLasy\mdsmwhtsquare$)\refstepcounter{mark}\ExpandArgs{c}\textsuperscript{@currentlabel}\Label{mark:orig:box}]
        { \Gamma'\braces[\big]{\brackets*{\UseLasy\mdwhtdiamond A\cramped{\sp\bot}, \Delta\braces*{\UseLasy\mdsmwhtsquare B}}} }
        \def\RuleName{$\UseLasy\mdsmwhtsquare$}
        \rewrite{%
          \raisebox\baselineskip{%
            \llap{%
              \ExpandArgs{c}\let{@currentlabelname}\RuleName%
              \refstepcounter{rule}\Label{tree:orig:box}%
            }%
          }%
          \box\treebox
        }
      \ellipsis{\small$\mathscr{D}_1$} { \Gamma\braces[\big]{\brackets*{\UseLasy\mdwhtdiamond A\cramped{\sp\bot}, A}} }
    \end{prooftree}
    \caption{Original derivation~$\mathscr{D}$.}\label{fig:orig}
  \end{subcaptionblock}\allowbreak
  \hfill
  \raisebox{3\baselineskip}{\begin{subcaptionblock}[t]{0.5\textwidth}\centering\ignorespaces
    \bigskip
    \begin{prooftree}
      \infer[dashed]0[(id)]   { \Gamma'\braces*{\brackets[\big]{\UseLasy\mdwhtdiamond B\cramped{\sp\bot}, \Delta\braces*{\brackets*{\UseLasy\mdwhtdiamond B\cramped{\sp\bot}, B, B\cramped{\sp\bot}}}}} }
      \infer1[($\UseLasy\mdwhtdiamond$)]   { \Gamma'\braces*{\brackets[\big]{\UseLasy\mdwhtdiamond B\cramped{\sp\bot}, \Delta\braces*{\brackets*{\UseLasy\mdwhtdiamond B\cramped{\sp\bot}, B}}}} }
      \infer1[($\UseLasy\mdsmwhtsquare$)]   { \Gamma'\braces*{\brackets[\big]{\UseLasy\mdwhtdiamond B\cramped{\sp\bot}, \Delta\braces*{\UseLasy\mdsmwhtsquare B}}} }
      \ellipsis{\small$\mathscr{D}_1'$} { \Gamma\braces[\big]{\brackets*{\UseLasy\mdwhtdiamond B\cramped{\sp\bot}, A}} }
    \end{prooftree}
    \caption{
      First step: Truncate the derivation
      above the application of~\mbox{\labelcref{tree:orig:box}}
      by adding~$\UseLasy\mdwhtdiamond B\cramped{\sp\bot}$ as an assumption.\label{fig:step:1st}
    }
  \end{subcaptionblock}}\bigskip
  
  \begin{subcaptionblock}[t]{\textwidth}\centering
    \begin{prooftree}
      \begingroup
        \set{rule style=no rule, template=}
        \begingroup
          \set{rule margin=0pt}
          \hypo{}
          \ellipsis{\small$\mathscr{D}_3$}{}
        \endgroup
      \endgroup
      \infer[no rule]1{ \Gamma''\braces*{ \brackets*{\UseLasy\mdwhtdiamond A\cramped{\sp\bot}, \Delta'\braces[\big]{\brackets*{\UseLasy\mdwhtdiamond B\cramped{\sp\bot}, \Theta, A\cramped{\sp\bot}}}} } }
      \begingroup
        \set{rule style=no rule, template=}
        \begingroup
          \set{rule margin=0pt}
          \hypo{}
          \ellipsis{\small\subref{fig:step:1st}}{}
        \endgroup
      \endgroup
      \infer[no rule]1{ \Gamma\braces[\big]{\brackets*{\UseLasy\mdwhtdiamond B\cramped{\sp\bot}, A}} }
      \infer[dashed]1[(weak)]   { \Gamma\braces*{\brackets*{\UseLasy\mdwhtdiamond B\cramped{\sp\bot}, A}, \brackets*{\UseLasy\mdwhtdiamond A\cramped{\sp\bot}, \Delta'\braces[\big]{\brackets*{\Theta}}}} }
      \infer[dashed]1[(rebase)] { \Gamma\braces*{\brackets*{\UseLasy\mdwhtdiamond A\cramped{\sp\bot}, \Delta'\braces[\big]{\brackets*{\UseLasy\mdwhtdiamond B\cramped{\sp\bot}, \Theta, A}}}} }
      \infer2[(cut)]    { \Gamma'''\braces*{\brackets*{\UseLasy\mdwhtdiamond A\cramped{\sp\bot}, \Delta'\braces[\big]{\brackets*{\UseLasy\mdwhtdiamond B\cramped{\sp\bot}, \Theta}}}} }
      \ellipsis{\small$\mathscr{D}_2' + \text{\labelcref*{tree:orig:box}} + \mathscr{D}_1''$}
        { \Gamma\braces[\big]{\brackets*{\UseLasy\mdwhtdiamond A\cramped{\sp\bot}, A}} }
    \end{prooftree}
    \caption{
      Second step:
      Displace the application of~\mbox{\labelcref{tree:orig:dia}}
      with a cut on~$A$.\label{fig:step:2nd}
    }
  \end{subcaptionblock}
  \par\medskip
  \begin{small}
    Here $\mathscr{D}_1'$, $\mathscr{D}_1''$, and $\mathscr{D}_2'$ denote
    minor modifications of $\mathscr{D}_1$, $\mathscr{D}_1$, and $\mathscr{D}_2$
    with admissible rules applied several times, respectively.
  \end{small}
  \caption{
    Overview of the diagonal-formula-elimination subprocedure.
    \label{fig:overview}
  }
\end{figure}
For simplicity, assume that \cref*{rule:dia} is not applied
with the~$\UseLasy\mdwhtdiamond A\cramped{\sp\bot}$ in~$\mathscr{D}_1$, $\mathscr{D}_2$, and~$\mathscr{D}_3$.
Then, to obtain~$\Gamma\braces{\lbrack A \rbrack}$, we need to erase the assumption~$A\cramped{\sp\bot}$
in the premise of~\labelcref*{tree:orig:dia}\footref{mark:orig:dia}
without using~\labelcref*{tree:orig:dia}.

Next, we consider how to erase $A\cramped{\sp\bot}$. This is done in the following steps.
First, truncate $\mathscr{D}$
above the application of~\labelcref*{tree:orig:box}\footref{mark:orig:box},
including the use of~\labelcref{tree:orig:dia},
by adding $\UseLasy\mdwhtdiamond B\cramped{\sp\bot}$ as an assumption~(\cref{fig:step:1st}).
Then return to the original~$\mathscr{D}$ and erase $A\cramped{\sp\bot}$
by a cut~(\cref{fig:step:2nd}), where many applications of admissible rules are
required to adjust the shape of sequents.
The added $\UseLasy\mdwhtdiamond B\cramped{\sp\bot}$ is dealt with as a diagonal formula and
does not remain in the conclusion.
The resulting derivation no longer requires the~$\UseLasy\mdwhtdiamond A\cramped{\sp\bot}$,
allowing us to obtain $\Gamma\braces{\lbrack A \rbrack}$.

This is the base case of our rewriting process, and in general,
it can be done by repeating this as many times as necessary.
However, it requires global manipulation of the derivation and
several tweaks of sub-derivations by admissible rules,
making precise discussion difficult.
In addition, it may seem a bit counterintuitive that
even if only \emph{two} instances of \cref*{tree:orig:dia} are involved,
at most \emph{three} cuts are needed.
To avoid pitfalls, we introduce \emph{annotations} in the next section
to allow for more precise arguments.

\section{Annotated System}\label{sec:annotation}

In this section, we introduce an auxiliary calculus with tiny annotations, which
keep track of the use of diagonal formulae in a derivation.

An \emph{annotated sequent} is defined by the following grammar:
\[
  \Gamma, \Delta \Coloneq \left.
    \mkern\medmuskip\mathord\cdot\vphantom{x}\mkern\medmuskip
  \mathrel{} \mathclose{} \big \vert \mathopen{} \mathrel{}
    \Gamma, C
  \mathrel{} \mathclose{} \big \vert \mathopen{} \mathrel{}
    \Gamma, \UseLasy\mdwhtdiamond A \sb{\Sigma}
  \mathrel{} \mathclose{} \big \vert \mathopen{} \mathrel{}
    \Gamma, \brackets{\Delta}_B
    \mkern1mu\text,
  \right.
\]
where $C$ is a formula not of the form~$\UseLasy\mdwhtdiamond A$, and $B$ and~$\Sigma$ are
a formula and a set of formulae, respectively.
Accordingly, there are two sorts of annotations, each with the following roles:
\begin{itemize}
  \item
    $\brackets{\mathord{-}}_B$ indicates that the~$\brackets{\mathord{-}}$ is to be discharged
    by applying \cref*{rule:box} to~$B$:
    \begin{prooftree*}
      \hypo    { \Gamma\braces[\big]{\brackets*{\UseLasy\mdwhtdiamond B\cramped{\sp\bot}, B}\cramped{_B}} }
      \infer1[($\UseLasy\mdsmwhtsquare$)] { \Gamma\braces[\big]{\UseLasy\mdsmwhtsquare B} }
    \end{prooftree*}
  \item
    The set~$\Sigma$ of~$\UseLasy\mdwhtdiamond A \Sb\Sigma$ records the provenances of $A$.
    That is, $B \in \Sigma$ implies that we have used \cref*{rule:dia} to~$A$
    directly contained inside some~$\brackets{\mathord{-}}_B$, absorbing it into the~$\UseLasy\mdwhtdiamond A$:
    \begin{prooftree*}
      \hypo { \Gamma\braces*{\Delta\braces[\big]{\brackets*{\Delta', A}\cramped{_B}}, \UseLasy\mdwhtdiamond A \cramped{\sb{\Sigma}} }}
      \infer1[($\UseLasy\mdwhtdiamond$)] { \Gamma\braces*{\Delta\braces[\big]{\brackets{\Delta'}\cramped{_B}}, \UseLasy{\mdwhtdiamond} A \cramped{\sb{\Sigma \medspace\mathord\cup\medspace\braces B}} }}
      \def\RuleName{$\UseLasy\mdwhtdiamond$}
      \rewrite{%
        \raisebox\baselineskip{%
          \llap{%
            \ExpandArgs{c}\let{@currentlabelname}\RuleName%
            \refstepcounter{rule}\Label{rule:annotated:dia}%
          }%
        }%
        \box\treebox
      }
    \end{prooftree*}
\end{itemize}
This is all our annotations do, and we shall see in the next section that
they do indeed provide sufficient information for induction.

To put it a little more strictly, for annotations to make sense,
we require the following conditions be placed on the inference rules:
\begin{itemize}
  \item
    An initial sequent shall contain only emptysets
    as an annotation to $\UseLasy\mdwhtdiamond$-formulae, since
    \cref*{rule:dia} has not yet been applied here.
    \begin{mathpar}
      \prescript{\raisebox{-1ex}{\hbox{\text{\small\faCheck}}}\enskip}{}{
        \begin{prooftree}
          \infer0[(id)] { \UseLasy\mdwhtdiamond A \Sb\emptyset, \brackets*{\alpha, \alpha\cramped{\sp\bot}, \brackets*{\UseLasy\mdwhtdiamond B \land C}\cramped{_{\alpha}}}\cramped{_{\beta}} }
        \end{prooftree}
      }
      \and
      \prescript{\raisebox{-1ex}{\hbox{\text{\faTimes}}}\enskip}{}{
        \begin{prooftree}
          \infer0[(id)] { \UseLasy\mdwhtdiamond A \Sb{\braces{\alpha}}, \brackets*{\alpha, \alpha\cramped{\sp\bot}, \brackets*{A}\cramped{_{\alpha}}}\cramped{_{\beta}} }
        \end{prooftree}
      }
    \end{mathpar}
  \item
    Even if A comes out of multiple $\brackets{\mathord{-}}$'s by \cref*{rule:dia},
    only the innermost one is essential.
    \begin{mathpar}
      \prescript{\raisebox{-1ex}{\hbox{\text{\small\faCheck}}}\enskip}{}{
        \begin{prooftree}
          \hypo    { \brackets[\big]{\brackets*{\Delta, A}\cramped{_{\alpha}}}\cramped{_{\beta}}, \UseLasy\mdwhtdiamond A \Sb{\Sigma} }
          \infer1[($\UseLasy\mdwhtdiamond$)] { \brackets[\big]{\brackets*{\Delta}\cramped{_{\alpha}}}\cramped{_{\beta}}, \UseLasy\mdwhtdiamond A \Sb{\Sigma \medspace\mathord\cup\medspace \braces{\alpha}} }
        \end{prooftree}
      }
      \and
      \prescript{\raisebox{-1ex}{\hbox{\text{\faTimes}}}\enskip}{}{
        \begin{prooftree}
          \hypo { \brackets[\big]{\brackets*{\Delta, A}\cramped{_{\alpha}}}\cramped{_{\beta}}, \UseLasy\mdwhtdiamond A \Sb{\Sigma} }
          \infer1[($\UseLasy\mdwhtdiamond$)]
            { \brackets[\big]{\brackets*{\Delta}\cramped{_{\alpha}}}\cramped{_{\beta}}, \UseLasy\mdwhtdiamond A \Sb{\Sigma \medspace\mathord\cup\medspace \braces{\alpha,\,\beta}} }
        \end{prooftree}
      }
      \and
      \prescript{\raisebox{-1ex}{\hbox{\text{\faTimes}}}\enskip}{}{
        \begin{prooftree}
          \hypo    { \brackets[\big]{\brackets*{\Delta, A}\cramped{_{\alpha}}}\cramped{_{\beta}}, \UseLasy\mdwhtdiamond A \Sb{\Sigma} }
          \infer1[($\UseLasy\mdwhtdiamond$)] { \brackets[\big]{\brackets*{\Delta}\cramped{_{\alpha}}}\cramped{_{\beta}}, \UseLasy\mdwhtdiamond A \Sb{\Sigma \medspace\mathord\cup\medspace \braces{\beta}} }
        \end{prooftree}
      }
    \end{mathpar}
  \item
    We never use annotations for subformulae\footnote{
      This is because in those cases the induction works just fine
      due to the smaller formula sizes.
    },
    so whenever logical rules are applied to $\UseLasy\mdwhtdiamond$-formulae,
    their annotations are simply discarded.
    \begin{mathpar}
      \prescript{\raisebox{-1ex}{\hbox{\text{\small\faCheck}}}\enskip}{}{
        \begin{prooftree}
          \hypo   { \UseLasy\mdwhtdiamond A \Sb{\Sigma}, \UseLasy\mdwhtdiamond B \Sb{\Pi} }
          \infer1[($\lor$)] { \UseLasy\mdwhtdiamond A \lor \UseLasy\mdwhtdiamond B }
        \end{prooftree}
      }
      \and
      \prescript{\raisebox{-1ex}{\hbox{\text{\faTimes}}}\enskip}{}{
        \begin{prooftree}
          \hypo   { \UseLasy\mdwhtdiamond A \Sb{\Sigma}, \UseLasy\mdwhtdiamond B \Sb{\Pi} }
          \infer1[($\lor$)] { \UseLasy\mdwhtdiamond A \Sb{\Sigma} \lor \UseLasy\mdwhtdiamond B \Sb{\Pi} }
        \end{prooftree}
      }
      \and
      \prescript{\raisebox{-1ex}{\hbox{\text{\faTimes}}}\enskip}{}{
        \begin{prooftree}
          \hypo    { \UseLasy\mdwhtdiamond A \Sb{\Sigma} }
          \hypo    { \UseLasy\mdwhtdiamond B \Sb{\Pi} }
          \infer2[($\land$)] { (\UseLasy\mdwhtdiamond A \land \UseLasy\mdwhtdiamond B) \Sb{\Sigma \medspace\mathord\cup\medspace \Pi} }
        \end{prooftree}
      }
      \\
      \prescript{\raisebox{-1ex}{\hbox{\text{\small\faCheck}}}\enskip}{}{
        \begin{prooftree}
          \hypo
            { \brackets[\big]{\Delta, \UseLasy\mdwhtdiamond A \Sb{\Sigma}}\cramped{_{\alpha}}, \UseLasy\mdwhtdiamond\UseLasy\mdwhtdiamond A \Sb{\Pi} }
          \infer1[($\UseLasy\mdwhtdiamond$)] { \brackets*{\Delta}\cramped{_{\alpha}}, \UseLasy\mdwhtdiamond \UseLasy\mdwhtdiamond A \Sb{\Pi \medspace\mathord\cup\medspace \braces{\alpha}} }
        \end{prooftree}
      }
      \and
      \prescript{\raisebox{-1ex}{\hbox{\text{\faTimes}}}\enskip}{}{
        \begin{prooftree}
          \hypo
            { \brackets[\big]{\Delta, \UseLasy\mdwhtdiamond A \Sb{\Sigma}}\cramped{_{\alpha}}, \UseLasy\mdwhtdiamond\UseLasy\mdwhtdiamond A \Sb{\Pi} }
          \infer1[($\UseLasy\mdwhtdiamond$)] { \brackets*{\Delta}\cramped{_{\alpha}}, \UseLasy\mdwhtdiamond \UseLasy\mdwhtdiamond A \Sb{\Sigma \medspace\mathord\cup\medspace \Pi \medspace\mathord\cup\medspace \braces{\alpha}} }
        \end{prooftree}
      }
    \end{mathpar}
  \item
    For the two-premise rules~(i.e.\@,~\labelcref*{rule:and,rule:cut}):
    \begin{itemize}
      \item
        For each~$\brackets{\mathord{-}}$, its annotation formula shall be
        \emph{shared} by both premises; whereas\samepage
      \item
        For each $\UseLasy\mdwhtdiamond$-formula,
        its annotation set may be \emph{independent} on two premises,
        and in the conclusion two sets are to be merged by set union~$\cup$.
    \end{itemize}
    \begin{mathpar}
      \prescript{\raisebox{-1ex}{\hbox{\text{\small\faCheck}}}\enskip}{}{
        \begin{prooftree}[rule separation=2em]
          \hypo    {A, \UseLasy\mdwhtdiamond C \Sb{\braces{E, F}}, \UseLasy\mdwhtdiamond D\cramped{\sp\bot} \Sb\emptyset}
          \hypo    {B, \UseLasy\mdwhtdiamond C \Sb{\braces{E}}, \UseLasy\mdwhtdiamond D\cramped{\sp\bot} \Sb{\braces{G}}}
          \infer2[($\land$)] {A \land B, \UseLasy\mdwhtdiamond C \Sb{\braces{E, F}}, \UseLasy\mdwhtdiamond D\cramped{\sp\bot} \Sb{\braces{G}}}
        \end{prooftree}
      }
    \end{mathpar}
    Since exchange may be used implicitly, it is not always possible to
    uniquely determined which formula of each premise is paired, but
    in that case anything is ok.
\end{itemize}
We may omit annotations if not important, and
note that a mere $\UseLasy\mdwhtdiamond A\cramped{\sp\bot}$ does not imply $\UseLasy\mdwhtdiamond A\cramped{\sp\bot} \Sb\emptyset$,
but rather $\UseLasy\mdwhtdiamond A\cramped{\sp\bot} \Sb{\Sigma}$ for some~(possibly empty)~$\Sigma$.

\begin{example}[The axiom~(K)]
  A proof of
  \begin{math}
    \UseLasy\mdsmwhtsquare (\alpha \to \beta) \to \UseLasy\mdsmwhtsquare \alpha \to \UseLasy\mdsmwhtsquare \beta \equiv
    \UseLasy\mdwhtdiamond \parens{\alpha \land \beta\cramped{\sp\bot}} \lor \parens{\UseLasy\mdwhtdiamond \alpha\cramped{\sp\bot} \lor \UseLasy\mdsmwhtsquare \beta}
  \end{math}
  is as follows:
  \begin{prooftree*}[strut=1.1]
    \small
    \infer0[(id)]  { \UseLasy\mdwhtdiamond \parens{\alpha \land \beta\cramped{\sp\bot}} \Sb\emptyset, \UseLasy\mdwhtdiamond \alpha\cramped{\sp\bot} \Sb\emptyset, \brackets[\big]{\UseLasy\mdwhtdiamond \beta\cramped{\sp\bot} \Sb\emptyset, \beta, \alpha\cramped{\sp\bot}, \alpha}\cramped{_{\beta}} }
    \infer1[($\UseLasy\mdwhtdiamond$)] { \UseLasy\mdwhtdiamond \parens{\alpha \land \beta\cramped{\sp\bot}} \Sb\emptyset, \UseLasy\mdwhtdiamond \alpha\cramped{\sp\bot} \Sb{\braces{\beta}}, \brackets[\big]{\UseLasy\mdwhtdiamond \beta\cramped{\sp\bot} \Sb\emptyset, \beta, \alpha}\cramped{_{\beta}} }
    \infer0[(id)]  { \UseLasy\mdwhtdiamond \parens{\alpha \land \beta\cramped{\sp\bot}} \Sb\emptyset, \UseLasy\mdwhtdiamond \alpha\cramped{\sp\bot} \Sb\emptyset, \brackets[\big]{\UseLasy\mdwhtdiamond \beta\cramped{\sp\bot} \Sb\emptyset, \beta, \beta\cramped{\sp\bot}}\cramped{_{\beta}} }
    \infer2[($\land$)] { \UseLasy\mdwhtdiamond \parens{\alpha \land \beta\cramped{\sp\bot}} \Sb\emptyset, \UseLasy\mdwhtdiamond \alpha\cramped{\sp\bot} \Sb{\braces{\beta}}, \brackets[\big]{\UseLasy\mdwhtdiamond \beta\cramped{\sp\bot} \Sb\emptyset, \beta, \alpha \land \beta\cramped{\sp\bot}}\cramped{_{\beta}} }
    \infer1[($\UseLasy\mdwhtdiamond$)] { \UseLasy\mdwhtdiamond \parens{\alpha \land \beta\cramped{\sp\bot}} \Sb{\braces{\beta}}, \UseLasy\mdwhtdiamond \alpha\cramped{\sp\bot} \Sb{\braces{\beta}}, \brackets[\big]{\UseLasy\mdwhtdiamond \beta\cramped{\sp\bot} \Sb\emptyset, \beta}\cramped{_{\beta}} }
    \infer1[($\UseLasy\mdsmwhtsquare$)] { \UseLasy\mdwhtdiamond \parens{\alpha \land \beta\cramped{\sp\bot}} \Sb{\braces{\beta}}, \UseLasy\mdwhtdiamond \alpha\cramped{\sp\bot} \Sb{\braces{\beta}}, \UseLasy\mdsmwhtsquare \beta }
    \infer1[($\lor$)]  { \UseLasy\mdwhtdiamond \parens{\alpha \land \beta\cramped{\sp\bot}} \Sb{\braces{\beta}}, \UseLasy\mdwhtdiamond \alpha\cramped{\sp\bot} \lor \UseLasy\mdsmwhtsquare \beta }
    \infer1[($\lor$)]  { \UseLasy\mdwhtdiamond \parens{\alpha \land \beta\cramped{\sp\bot}} \lor \parens{\UseLasy\mdwhtdiamond \alpha\cramped{\sp\bot} \lor \UseLasy\mdsmwhtsquare \beta} }
    \AddToHookNext{env/equation*/end}{\qedhere}
  \end{prooftree*}
\end{example}

We summarize the annotated system in \cref{fig:annotated}.
\begin{figure}[t]
  \begin{mathpar}
    \begin{prooftree}
      \hypo \mathstrut
      \infer1[(id)] { \Gamma\braces*{\alpha\cramped{\sp\bot}, \alpha} }
    \end{prooftree}
    \footref{cond:initial} \and
    \begin{prooftree}
      \hypo    { \Gamma\braces*{A} }
      \hypo    { \Gamma\braces*{B} }
      \infer2[($\land$)] { \Gamma\braces*{A \land B} }
    \end{prooftree}
    \multfootref{cond:subformula;cond:context} \and
    \begin{prooftree}
      \hypo   { \Gamma\braces*{A, B} }
      \infer1[($\lor$)] { \Gamma\braces*{A \lor B} }
    \end{prooftree}
    \footref{cond:subformula}
    \\
    \begin{prooftree}
      \hypo    { \Gamma\braces[\big]{\brackets*{\UseLasy\mdwhtdiamond A\cramped{\sp\bot} \Sb{\Sigma}, A}\cramped{_A}} }
      \infer1[($\UseLasy\mdsmwhtsquare$)] { \Gamma\braces*{\UseLasy\mdsmwhtsquare A} }
      \def\RuleName{$\UseLasy\mdsmwhtsquare$}
      \rewrite{%
        \raisebox\baselineskip{%
          \llap{%
            \ExpandArgs{c}\let{@currentlabelname}\RuleName%
            \refstepcounter{rule}\Label{rule:annotated:box}%
          }%
        }%
        \box\treebox
      }
    \end{prooftree}
    \footref{cond:subformula} \and
    \begin{prooftree}
      \hypo { \Gamma\braces*{\Delta\braces[\big]{\brackets*{\Delta', A}\cramped{_B}}, \UseLasy\mdwhtdiamond A \Sb{\Sigma}} }
      \infer1[($\UseLasy\mdwhtdiamond$)] { \Gamma\braces*{\Delta\braces[\big]{\brackets{\Delta'}\cramped{_B}}, \UseLasy\mdwhtdiamond A \Sb{\Sigma \medspace\mathord\cup\medspace\braces B}} }
    \end{prooftree}
    \footref{cond:subformula} \and
    \begin{prooftree}
      \hypo    { \Gamma\braces*{A} }
      \hypo    { \Gamma\braces*{A\cramped{\sp\bot}} }
      \infer2[(cut)] { \Gamma\braces*{} }
    \end{prooftree}
    \footref{cond:context}
  \end{mathpar}
  \begin{footnotesize}%
    \begin{tasks}[
      label-format=\alph{task}\rparen\hskip0.8\labelsep\UseName{@gobble},
      ref=\textasciiasterisk\alph*,
      item-format=
        \vskip-\baselineskip\nointerlineskip 
        \AddToHookNext{para/begin}{\strut}\vtop,
      after-skip=0pt,
    ](2)
      \task\label{cond:initial}
        $\Gamma\braces{\mathord{-}}$ must contain only $\emptyset$ as annotation sets.
      \task\label{cond:subformula}
        Discard the annotation(s) of $A$ (and $B$) if exist(s).
      \task*\label{cond:context}
        Each $\brackets{\mathord{-}}$ must have the same annotation for premises.
        For each $\UseLasy\mdwhtdiamond$-formula, annotation sets are merged in the conclusion.
    \end{tasks}%
  \end{footnotesize}
  \caption{Inference rules with annotations.}\label{fig:annotated}
\end{figure}
Observe that the rules in \cref{fig:annotated}, with all annotations dropped,
are exactly the same as those in \cref{fig:rules,def:cut},
including in particular \cref{rule:annotated:dia}.
Also note that annotations are not a restriction, but merely clues,
because, given an unannotated derivation,
we can always lift it straightforwardly to an annotated one;
more precisely, this is done by the following two steps:
\begin{enumerate}
  \item
    First, we need to annotate all $\brackets{\mathord{-}}$'s appropriately.
    Look at a given derivation from bottom to top, and
    if we find that \cref*{rule:box} is applied to some $A$, then
    annotate $\brackets{\mathord{-}}$ there with $A$.
    Taking the derivation of the Löb axiom~(\cpageref{example:löb}) as an example,
    there are two instances, annotated as follows:
    \begin{align*}
      \begin{prooftree}
        \begingroup
          \set{rule style=no rule, template=}
          \begingroup
            \set{rule margin=0pt}
            \hypo{}
            \ellipsis{}{}
          \endgroup
        \endgroup
        \infer[no rule]1{ \UseLasy\mdwhtdiamond \parens*{\UseLasy\mdsmwhtsquare \alpha \land \alpha\cramped{\sp\bot}}, \brackets*{\UseLasy\mdwhtdiamond \alpha\cramped{\sp\bot}, \alpha} }
        \infer1[($\UseLasy\mdsmwhtsquare$)] { \UseLasy\mdwhtdiamond \parens*{\UseLasy\mdsmwhtsquare \alpha \land \alpha\cramped{\sp\bot}}, \UseLasy\mdsmwhtsquare \alpha }
        \begingroup
          \set{rule style=no rule, template=}
          \infer1{}
          \begingroup
            \set{rule margin=0pt}
            \ellipsis{}{}
          \endgroup
        \endgroup
      \end{prooftree}
      &\quad\tikz[baseline={(0,-0.3ex)}] \pic[anchor=center]{leadsto};\quad
      \begin{prooftree}
        \begingroup
          \set{rule style=no rule, template=}
          \begingroup
            \set{rule margin=0pt}
            \hypo{}
            \ellipsis{}{}
          \endgroup
        \endgroup
        \infer[no rule]1{ \UseLasy\mdwhtdiamond \parens*{\UseLasy\mdsmwhtsquare \alpha \land \alpha\cramped{\sp\bot}}, \brackets*{\UseLasy\mdwhtdiamond \alpha\cramped{\sp\bot}, \alpha}\cramped{_{\pmb{\alpha}}} }
        \infer1[($\UseLasy\mdsmwhtsquare$)] { \UseLasy\mdwhtdiamond \parens*{\UseLasy\mdsmwhtsquare \alpha \land \alpha\cramped{\sp\bot}}, \UseLasy\mdsmwhtsquare \alpha }
        \begingroup
          \set{rule style=no rule, template=}
          \infer1{}
          \begingroup
            \set{rule margin=0pt}
            \ellipsis{}{}
          \endgroup
        \endgroup
      \end{prooftree}
      \\
      \begin{prooftree}
        \begingroup
          \set{rule style=no rule, template=}
          \begingroup
            \set{rule margin=0pt}
            \hypo{}
            \ellipsis{}{}
          \endgroup
        \endgroup
        \infer[no rule]1{ \UseLasy\mdwhtdiamond \parens*{\UseLasy\mdsmwhtsquare \alpha \land \alpha\cramped{\sp\bot}}, \brackets[\big]{\UseLasy\mdwhtdiamond \alpha\cramped{\sp\bot}, \alpha, \brackets*{\UseLasy\mdwhtdiamond \alpha\cramped{\sp\bot}, \alpha}}\cramped{_{\alpha}} }
        \infer1[($\UseLasy\mdsmwhtsquare$)] { \UseLasy\mdwhtdiamond \parens*{\UseLasy\mdsmwhtsquare \alpha \land \alpha\cramped{\sp\bot}}, \brackets*{\UseLasy\mdwhtdiamond \alpha\cramped{\sp\bot}, \alpha, \UseLasy\mdsmwhtsquare \alpha}\cramped{_{\alpha}} }
        \begingroup
          \set{rule style=no rule, template=}
          \infer1{}
          \begingroup
            \set{rule margin=0pt}
            \ellipsis{}{}
          \endgroup
        \endgroup
      \end{prooftree}
      &\quad\tikz[baseline={(0,-0.3ex)}] \pic[anchor=center]{leadsto};\quad
      \begin{prooftree}
        \begingroup
          \set{rule style=no rule, template=}
          \begingroup
            \set{rule margin=0pt}
            \hypo{}
            \ellipsis{}{}
          \endgroup
        \endgroup
        \infer[no rule]1{ \UseLasy\mdwhtdiamond \parens*{\UseLasy\mdsmwhtsquare \alpha \land \alpha\cramped{\sp\bot}}, \brackets[\big]{\UseLasy\mdwhtdiamond \alpha\cramped{\sp\bot}, \alpha, \brackets*{\UseLasy\mdwhtdiamond \alpha\cramped{\sp\bot}, \alpha}\cramped{_{\pmb{\alpha}}}}\cramped{_{\alpha}} }
        \infer1[($\UseLasy\mdsmwhtsquare$)] { \UseLasy\mdwhtdiamond \parens*{\UseLasy\mdsmwhtsquare \alpha \land \alpha\cramped{\sp\bot}}, \brackets*{\UseLasy\mdwhtdiamond \alpha\cramped{\sp\bot}, \alpha, \UseLasy\mdsmwhtsquare \alpha}\cramped{_{\alpha}} }
        \begingroup
          \set{rule style=no rule, template=}
          \infer1{}
          \begingroup
            \set{rule margin=0pt}
            \ellipsis{}{}
          \endgroup
        \endgroup
      \end{prooftree}
    \end{align*}
    If the end-sequent contains $\brackets{\mathord{-}}$'s,
    their annotations are not important and may be annotated in any way.
  \item
    Then, annotate all $\UseLasy\mdwhtdiamond$-formulae in an initial sequent with an emptyset
    and, from top to bottom, collect their usage.
    In the case of the example (\cpageref{example:löb}):
    \begin{align*}
      \begin{prooftree}
        \footnotesize
        \infer0[(id)]  { \UseLasy\mdwhtdiamond \parens*{\UseLasy\mdsmwhtsquare \alpha \land \alpha\cramped{\sp\bot}}, \brackets[\big]{\UseLasy\mdwhtdiamond \alpha\cramped{\sp\bot}, \alpha, \brackets*{\UseLasy\mdwhtdiamond \alpha\cramped{\sp\bot}, \alpha, \alpha\cramped{\sp\bot}}\cramped{_{\alpha}}}\cramped{_{\alpha}} }
        \infer1[($\UseLasy\mdwhtdiamond$)] { \UseLasy\mdwhtdiamond \parens*{\UseLasy\mdsmwhtsquare \alpha \land \alpha\cramped{\sp\bot}}, \brackets[\big]{\UseLasy\mdwhtdiamond \alpha\cramped{\sp\bot}, \alpha, \brackets*{\UseLasy\mdwhtdiamond \alpha\cramped{\sp\bot}, \alpha}\cramped{_{\alpha}}}\cramped{_{\alpha}} }
        \begingroup
          \set{rule style=no rule, template=}
          \infer1{}
          \begingroup
            \set{rule margin=0pt}
            \ellipsis{}{}
          \endgroup
        \endgroup
      \end{prooftree}
      &\quad\tikz[baseline={(0,-0.3ex)}] \pic[anchor=center]{leadsto};\quad
      \begin{prooftree}
        \footnotesize
        \infer0[(id)]  {
          \UseLasy\mdwhtdiamond \parens*{\UseLasy\mdsmwhtsquare \alpha \land \alpha\cramped{\sp\bot}} \Sb{\pmb{\emptyset}},
          \brackets[\big]{
            \UseLasy\mdwhtdiamond \alpha\cramped{\sp\bot} \Sb{\pmb{\emptyset}}, \alpha,
            \brackets*{\UseLasy\mdwhtdiamond \alpha\cramped{\sp\bot} \Sb{\pmb{\emptyset}}, \alpha, \alpha\cramped{\sp\bot}}\cramped{_{\alpha}}
          }\cramped{_{\alpha}}
        }
        \infer1[($\UseLasy\mdwhtdiamond$)] {
          \UseLasy\mdwhtdiamond \parens*{\UseLasy\mdsmwhtsquare \alpha \land \alpha\cramped{\sp\bot}} \Sb{\pmb{\emptyset}},
          \brackets[\big]{\UseLasy\mdwhtdiamond \alpha\cramped{\sp\bot} \Sb{\pmb{\braces{\alpha}}}, \alpha,
          \brackets*{\UseLasy\mdwhtdiamond \alpha\cramped{\sp\bot} \Sb{\pmb{\emptyset}}, \alpha}\cramped{_{\alpha}}}\cramped{_{\alpha}}
        }
        \begingroup
          \set{rule style=no rule, template=}
          \infer1{}
          \begingroup
            \set{rule margin=0pt}
            \ellipsis{}{}
          \endgroup
        \endgroup
      \end{prooftree}
      \\[\smallskipamount]
      \begin{prooftree}
        \small
        \begingroup
          \set{rule style=no rule, template=}
          \begingroup
            \set{rule margin=0pt}
            \hypo{}
            \ellipsis{}{}
          \endgroup
        \endgroup
        \infer[no rule]1{ \UseLasy\mdwhtdiamond \parens*{\UseLasy\mdsmwhtsquare \alpha \land \alpha\cramped{\sp\bot}} \Sb\emptyset, \brackets*{\UseLasy\mdwhtdiamond \alpha\cramped{\sp\bot} \Sb{\braces{\alpha}}, \alpha, \UseLasy\mdsmwhtsquare \alpha \land \alpha\cramped{\sp\bot}}\cramped{_{\alpha}} }
        \infer1[($\UseLasy\mdwhtdiamond$)] { \UseLasy\mdwhtdiamond \parens*{\UseLasy\mdsmwhtsquare \alpha \land \alpha\cramped{\sp\bot}}, \brackets*{\UseLasy\mdwhtdiamond \alpha\cramped{\sp\bot}, \alpha}\cramped{_{\alpha}} }
        \begingroup
          \set{rule style=no rule, template=}
          \infer1{}
          \begingroup
            \set{rule margin=0pt}
            \ellipsis{}{}
          \endgroup
        \endgroup
      \end{prooftree}
      &\quad\tikz[baseline={(0,-0.3ex)}] \pic[anchor=center]{leadsto};\quad
      \begin{prooftree}
        \small
        \begingroup
          \set{rule style=no rule, template=}
          \begingroup
            \set{rule margin=0pt}
            \hypo{}
            \ellipsis{}{}
          \endgroup
        \endgroup
        \infer[no rule]1{ \UseLasy\mdwhtdiamond \parens*{\UseLasy\mdsmwhtsquare \alpha \land \alpha\cramped{\sp\bot}} \Sb\emptyset, \brackets*{\UseLasy\mdwhtdiamond \alpha\cramped{\sp\bot} \Sb{\braces{\alpha}}, \alpha, \UseLasy\mdsmwhtsquare \alpha \land \alpha\cramped{\sp\bot}}\cramped{_{\alpha}} }
        \infer1[($\UseLasy\mdwhtdiamond$)] {
          \UseLasy\mdwhtdiamond \parens*{\UseLasy\mdsmwhtsquare \alpha \land \alpha\cramped{\sp\bot}} \Sb{\pmb{\braces{\alpha}}}, \brackets*{\UseLasy\mdwhtdiamond \alpha\cramped{\sp\bot} \Sb{\pmb{\braces{\alpha}}}, \alpha}\cramped{_{\alpha}}
        }
        \begingroup
          \set{rule style=no rule, template=}
          \infer1{}
          \begingroup
            \set{rule margin=0pt}
            \ellipsis{}{}
          \endgroup
        \endgroup
      \end{prooftree}
    \end{align*}
\end{enumerate}

\begin{remark}
  The use of \emph{sets} to annotate $\UseLasy\mdwhtdiamond$-formulae is
  not essential in our proof.
  In fact, it is possible to use \emph{multisets} instead, but using sets makes
  the proof a bit simpler since we can deal with identical formulae at once
  (see \cref{claim:step:1st,claim:step:2nd}).
\end{remark}

Hereafter, we shall focus on a fixed $\UseLasy\mdwhtdiamond A\cramped{\sp\bot}$ and
treat a cut on $A$ as a rather first-class inference rule.
We write $\mathord\vdash\penalty\binoppenalty\mskip\thickmuskip \Gamma$ if the sequent $\Gamma$ is cut-free provable, and
$\mathord\vdash\cramped{\sp{A}}\penalty\binoppenalty\mskip\thickmuskip \Gamma$ if provable with cuts on $A$ admitted.
It is easily checked that the admissible rules
in \cref{claim:identity,claim:admissible} are still admissible
under proper annotations, with or without cuts;
in particular, the following forms are admissible:
\begin{mathpar}
  \ebproofset{center=false}
  \begin{prooftree}
    \infer0[(id)] { \mathord\vdash\Sp{A}\penalty\binoppenalty\mskip\thickmuskip \Gamma\braces*{\UseLasy\mdwhtdiamond A\cramped{\sp\bot} \Sb\emptyset}\braces*{B\cramped{\sp\bot}, B} }
  \end{prooftree}
  \and
  \begin{prooftree}
    \hypo     { \mathord\vdash\Sp{A}\penalty\binoppenalty\mskip\thickmuskip \Gamma\braces*{} }
    \infer1[(weak)] { \mathord\vdash\Sp{A}\penalty\binoppenalty\mskip\thickmuskip \Gamma\braces*{\UseLasy\mdwhtdiamond A\cramped{\sp\bot} \Sb\emptyset} }
  \end{prooftree}
  \and
  \begin{prooftree}
    \hypo         { \mathord\vdash\Sp{A}\penalty\binoppenalty\mskip\thickmuskip \Gamma\braces*{\UseLasy\mdwhtdiamond A\cramped{\sp\bot} \Sb{\Sigma}, \UseLasy\mdwhtdiamond A\cramped{\sp\bot} \Sb{\Sigma'}} }
    \infer1[(contract)] { \mathord\vdash\Sp{A}\penalty\binoppenalty\mskip\thickmuskip \Gamma\braces*{\UseLasy\mdwhtdiamond A\cramped{\sp\bot} \Sb{\Sigma \medspace\mathord\cup\medspace \Sigma'}} }
  \end{prooftree}
\end{mathpar}

\begin{lemma}
  The following rule is admissible:
  \begin{prooftree*}[right label template=\rlap{\normalfont\inserttext}]
    \hypo            { \mathord\vdash\Sp{A}\penalty\binoppenalty\mskip\thickmuskip \Gamma\braces[\big]{\Delta\braces*{\UseLasy\mdwhtdiamond A\cramped{\sp\bot} \Sb{\Sigma}}} }
    \infer1[(cherry-pick)] { \mathord\vdash\Sp{A}\penalty\binoppenalty\mskip\thickmuskip \Gamma\braces[\big]{\Delta\braces*{}, \UseLasy\mdwhtdiamond A\cramped{\sp\bot} \Sb{\Sigma}} }
  \end{prooftree*}
\end{lemma}

\begin{proof}
  By induction on derivation.
\end{proof}

\section{Diagonal-Formula-Elimination}\label{sec:diagonal-elim}

In this section, we prove the key \lcnamecref{claim:main} of this paper
by fully using annotations. Our goal is to show the following:

\begin{lemma}[Diagonal-formula-elimination]\label{claim:main}
  If\/ $\mathord\vdash\penalty\binoppenalty\mskip\thickmuskip \Gamma\braces{\brackets{\UseLasy\mdwhtdiamond A\cramped{\sp\bot}, A}}$, then\/ $\mathord\vdash\cramped{\sp{A}}\penalty\binoppenalty\mskip\thickmuskip \Gamma\braces{\brackets{A}}$.
\end{lemma}

We first show that in exchange for adding an extra assumption,
we can reduce annotated formulae from~$\UseLasy\mdwhtdiamond A\cramped{\sp\bot}$.
This is similar to the step shown in \cref{fig:step:1st}, but
we process all $\brackets{\mathord{-}}$'s annotated with the same formula at once.

\begin{lemma}\label{claim:step:1st}
  Suppose\/ $\mathord\vdash\cramped{\sp{A}}\penalty\binoppenalty\mskip\thickmuskip \Gamma\braces{\brackets{\UseLasy\mdwhtdiamond A\cramped{\sp\bot} \Sb{\Sigma}, \Delta}}$, and let $B \in \Sigma$.
  If $\Delta$ contains no\/ $\brackets{\mathord{-}}$ annotated with $B$,
  then\/ $\mathord\vdash\Sp{A}\penalty\binoppenalty\mskip\thickmuskip \Gamma\braces{\brackets{\UseLasy\mdwhtdiamond A\cramped{\sp\bot} \Sb{\Sigma'} \mathcomma\penalty\binoppenalty \UseLasy\mdwhtdiamond B\cramped{\sp\bot} \mathcomma\penalty\binoppenalty \Delta}}$
  for some\/ $\Sigma' \subseteq \Sigma \setminus\braces B$.
\end{lemma}

\begin{proof}
  By induction on derivation. Here we show only two cases.
  \begin{case}
    The most important case is where
    \cref*{rule:annotated:box} is applied to $B$ within $\Delta$:
    \begin{prooftree*}[right label template=\rlap{\normalfont\inserttext}]
      \begingroup
        \set{rule style=no rule, template=}
        \begingroup
          \set{rule margin=0pt}
          \hypo{}
          \ellipsis{\small$\mathscr{D}$}{}
        \endgroup
      \endgroup
      \infer[no rule]1{ \Gamma\braces*{\brackets[\big]{\UseLasy\mdwhtdiamond A\cramped{\sp\bot} \Sb{\Sigma}, \Delta'\braces*{\brackets*{\UseLasy\mdwhtdiamond B\cramped{\sp\bot}, B}\cramped{_B}}}} }
      \infer1[($\UseLasy\mdsmwhtsquare$)] { \Gamma\braces*{\brackets[\big]{\UseLasy\mdwhtdiamond A\cramped{\sp\bot} \Sb{\Sigma}, \Delta'\braces*{\UseLasy\mdsmwhtsquare B}}} }
    \end{prooftree*}
    We cannot then use the induction hypothesis
    because the newly appering~$\brackets{\mathord{-}}_B$ in the premise
    breaks the requirement of the claim.
    So instead, we derive the desired sequent without using~$\mathscr{D}$ as follows:
    \begin{prooftree*}[right label template=\rlap{\normalfont\inserttext}]
      \infer[dashed]0[(id)] { \Gamma\braces*{\brackets[\big]{\UseLasy\mdwhtdiamond A\cramped{\sp\bot} \Sb\emptyset, \UseLasy\mdwhtdiamond B\cramped{\sp\bot}, \Delta'\braces*{\brackets*{\UseLasy\mdwhtdiamond B\cramped{\sp\bot}, B, B\cramped{\sp\bot}}\cramped{_B}}}} }
      \infer1[($\UseLasy\mdwhtdiamond$)] { \Gamma\braces*{\brackets[\big]{\UseLasy\mdwhtdiamond A\cramped{\sp\bot} \Sb\emptyset, \UseLasy\mdwhtdiamond B\cramped{\sp\bot}, \Delta'\braces*{\brackets*{\UseLasy\mdwhtdiamond B\cramped{\sp\bot}, B}\cramped{_B}}}} }
      \infer1[($\UseLasy\mdsmwhtsquare$)] { \Gamma\braces*{\brackets[\big]{\UseLasy\mdwhtdiamond A \cramped{\sp\bot}\Sb\emptyset, \UseLasy\mdwhtdiamond B\cramped{\sp\bot}, \Delta'\braces*{\UseLasy\mdsmwhtsquare B}}} }
    \end{prooftree*}
  \end{case}
  \begin{case}
    Suppose
    \begin{prooftree*}[right label template=\rlap{\normalfont\inserttext}]
      \begingroup
        \set{rule style=no rule, template=}
        \begingroup
          \set{rule margin=0pt}
          \hypo{}
          \ellipsis{\small$\mathscr{D}$}{}
        \endgroup
      \endgroup
      \infer[no rule]1{ \Gamma\braces*{\brackets[\big]{\UseLasy\mdwhtdiamond A\cramped{\sp\bot} \Sb{\Sigma}, \Delta'\braces*{A\cramped{\sp\bot}}}} }
      \infer1[($\UseLasy\mdwhtdiamond$)] { \Gamma\braces*{\brackets[\big]{\UseLasy\mdwhtdiamond A\cramped{\sp\bot} \Sb{\Sigma\medspace\mathord\cup\medspace\braces C}, \Delta'\braces*{}}} }
    \end{prooftree*}
    By assumption, we have $C \not\equiv B$; otherwise, there exists
    a~$\brackets{\mathord{-}}_B$ within~$\Delta'\braces{\mathord{-}}$, a~contradiction.
    Hence $B \in \Sigma$, and
    using the induction hypothesis we have the following:
    \begin{prooftree*}[right label template=\rlap{\normalfont\inserttext}]
      \begingroup
        \set{rule style=no rule, template=}
        \begingroup
          \set{rule margin=0pt}
          \hypo{}
          \ellipsis{\small$\mathscr{D}$}{}
        \endgroup
      \endgroup
      \infer[no rule]1{ \Gamma\braces*{\brackets[\big]{\UseLasy\mdwhtdiamond A\cramped{\sp\bot} \Sb{\Sigma}, \Delta'\braces*{A\cramped{\sp\bot}}}} }
      \infer1[(IH)]  { \Gamma\braces*{\brackets[\big]{\UseLasy\mdwhtdiamond A\cramped{\sp\bot} \Sb{\Sigma'}, \UseLasy\mdwhtdiamond B\cramped{\sp\bot}, \Delta'\braces*{A\cramped{\sp\bot}}}} }
      \infer1[($\UseLasy\mdwhtdiamond$)] { \Gamma\braces*{\brackets[\big]{\UseLasy\mdwhtdiamond A\cramped{\sp\bot} \Sb{\Sigma' \medspace\mathord\cup\medspace\braces C}, \UseLasy\mdwhtdiamond B\cramped{\sp\bot}, \Delta'\braces*{}}} }
    \end{prooftree*}
    where $\Sigma' \subseteq \Sigma \setminus\braces B$, and so
    $\Sigma' \medspace\mathord\cup\medspace\braces C \subseteq (\Sigma \medspace\mathord\cup\medspace\braces C) \setminus\braces B$ as required.
  \end{case}
  The other cases are straightforward.
\end{proof}

\begin{definition}
  Let $\Delta$ be an annotated sequent.
  Then, the sequent~$\Delta\Sp{+}$ is defined inductively as follows:
  \begin{align*}
    (\mkern\medmuskip\mathord\cdot\vphantom{x}\mkern\medmuskip)\Sp{+} &\equiv \mkern\medmuskip\mathord\cdot\vphantom{x}\mkern\medmuskip \mkern1mu\text; \\
    \parens[\big]{\Gamma, C}\Sp{+} &\equiv \Gamma\Sp{+}, C \mkern1mu\text; \\
    \parens[\big]{\Gamma, \brackets*{\Delta}\cramped{_C}}\Sp{+} &\equiv
      \Gamma\Sp{+}, \brackets*{\Delta\Sp{+}, \UseLasy\mdwhtdiamond C\cramped{\sp\bot}}\cramped{_C} \mkern1mu\text.
  \end{align*}
  The weakened context $\Delta\Sp{+}\braces{\mathord{-}}$ for $\Delta\braces{\mathord{-}}$ is
  defined in a similar way.
\end{definition}

In short, $\Delta\Sp{+}$ is the sequent in which all subsequents
of the form~$\brackets{\Delta'}\cramped{_C}$ within $\Delta$ are weakened to
$\brackets{\Delta' \mathcomma\penalty\binoppenalty \UseLasy\mdwhtdiamond C\cramped{\sp\bot}}\cramped{_C}$ with the diagonal formula~$\UseLasy\mdwhtdiamond C\cramped{\sp\bot}$.
For example,
$\parens[\big]{A \mathcomma\penalty\binoppenalty \brackets[\big]{\brackets{B}\cramped{_B} \mathcomma\penalty\binoppenalty \brackets{C \mathcomma\penalty\binoppenalty D}\cramped{_E}}\cramped{_C}}\sp{+}$
represents the sequent
\[
  A, \brackets[\big]{\brackets*{B, \UseLasy\mdwhtdiamond B\cramped{\sp\bot}}\cramped{_B}, \brackets*{C, D, \UseLasy\mdwhtdiamond E\cramped{\sp\bot}}\cramped{_E}, \UseLasy\mdwhtdiamond C\cramped{\sp\bot}}\cramped{_C} \mkern1mu\text.
\]

We next show that the extra assumption~$\UseLasy\mdwhtdiamond B\cramped{\sp\bot}$ added in \cref{claim:step:1st}
can be moved to its \textquote{proper place}, namely, inside some $\brackets{\mathord{-}}_B$,
where it should eventually be contracted with the diagonal formula.
This corresponds to the step of \cref{fig:step:2nd}, but again
we handle multiple instances simultaneously.

\begin{lemma}\label{claim:step:2nd}
  Suppose
  \[
    \mathord\vdash\sp{A}\penalty\binoppenalty\mskip\thickmuskip \Gamma\braces[\big]{\brackets*{\UseLasy\mdwhtdiamond A\cramped{\sp\bot} \sb{\Sigma}, \UseLasy\mdwhtdiamond B\cramped{\sp\bot}, A}} \mkern1mu\text.
    \tag{\dag1} \label{eq:dag:1}
  \]
  If
  \[
    \mathord\vdash\sp{A}\penalty\binoppenalty\mskip\thickmuskip \Gamma\braces*{\brackets*{\UseLasy\mdwhtdiamond A\cramped{\sp\bot} \sb{\Pi}, \Delta}} \mkern1mu\text,
    \tag{\dag2} \label{eq:dag:2}
  \]
  then\/ $\mathord\vdash\Sp{A}\penalty\binoppenalty\mskip\thickmuskip \Gamma\braces{\brackets{\UseLasy\mdwhtdiamond A\cramped{\sp\bot} \Sb{\Pi'}, \Delta\Sp{+}}}$
  for some\/ $\Pi' \subseteq \Sigma \cup (\Pi \setminus\braces B)$.
\end{lemma}

\begin{proof}
  By induction on the derivation of~\eqref{eq:dag:2}.
  Here we show only two important cases.
  \begin{case}
    Suppose
    \begin{prooftree*}[right label template=\rlap{\normalfont\inserttext}]
      \begingroup
        \set{rule style=no rule, template=}
        \begingroup
          \set{rule margin=0pt}
          \hypo{}
          \ellipsis{\small$\mathscr{D}$}{}
        \endgroup
      \endgroup
      \infer[no rule]1{ \Gamma\braces*{\brackets*{\UseLasy\mdwhtdiamond A\cramped{\sp\bot} \Sb{\Pi}, \Delta'\braces[\big]{\brackets*{\Delta'', A\cramped{\sp\bot}}\cramped{_B}}}} }
      \infer1[($\UseLasy\mdwhtdiamond$)] { \Gamma\braces*{\brackets*{\UseLasy\mdwhtdiamond A\cramped{\sp\bot} \Sb{\Pi \medspace\mathord\cup\medspace\braces B}, \Delta'\braces[\big]{\brackets*{\Delta''}\cramped{_B}}}} }
    \end{prooftree*}
    To avoid $B$ appended to $\UseLasy\mdwhtdiamond A\cramped{\sp\bot}$ again,
    we dispense with~\labelcref*{rule:annotated:dia} using a cut as follows:
    \begin{prooftree*}
      \small
      \begingroup
        \set{rule style=no rule, template=}
        \begingroup
          \set{rule margin=0pt}
          \hypo{}
          \ellipsis{\small$\mathscr{D}$}{}
        \endgroup
      \endgroup
      \infer[no rule]1{ \Gamma\braces*{\brackets*{\UseLasy\mdwhtdiamond A\cramped{\sp\bot} \Sb{\Pi}, \Delta'\braces[\big]{\brackets*{\Delta'', A\cramped{\sp\bot}}\cramped{_B}}}} }
      \infer1[(IH)] {
        \Gamma\braces[\big]{
          \brackets*{\UseLasy\mdwhtdiamond A\cramped{\sp\bot} \Sb{\Pi'}, \Delta'\Sp{+}\braces*{\brackets*{\Delta''\Sp{+}, \UseLasy\mdwhtdiamond B\cramped{\sp\bot}, A\cramped{\sp\bot}}\cramped{_B}}}
        }
      }
      \infer0[(\ref{eq:dag:1})] { \Gamma\braces*{\brackets*{\UseLasy\mdwhtdiamond A\cramped{\sp\bot} \Sb{\Sigma}, \UseLasy\mdwhtdiamond B\cramped{\sp\bot}, A}} }
      \infer[dashed]1[(weak)] { \Gamma\braces*{\brackets[\big]{\Delta'\Sp{+}\braces*{\brackets*{\Delta''\Sp{+}}\cramped{_B}}}, \brackets*{\UseLasy\mdwhtdiamond A\cramped{\sp\bot} \Sb{\Sigma}, \UseLasy\mdwhtdiamond B\cramped{\sp\bot}, A}} }
      \infer[dashed]1[(rebase)]
        { \Gamma\braces*{\brackets[\big]{\Delta'\Sp{+}\braces[\big]{\brackets*{\Delta''\Sp{+}, \UseLasy\mdwhtdiamond A\cramped{\sp\bot} \Sb{\Sigma}, \UseLasy\mdwhtdiamond B\cramped{\sp\bot}, A}\cramped{_B}}}} }
      \set{right label template=\rlap{\normalfont\inserttext}}
      \infer[dashed]1[(cherry-pick)]
        { \Gamma\braces*{ \brackets*{\UseLasy\mdwhtdiamond A\cramped{\sp\bot} \Sb{\Sigma}, \Delta'\Sp{+}\braces[\big]{\brackets*{\Delta''\Sp{+}, \UseLasy\mdwhtdiamond B\cramped{\sp\bot}, A}\cramped{_B}}} } }
      \infer2[(cut)]
        { \Gamma\braces*{ \brackets[\big]{\UseLasy\mdwhtdiamond A\cramped{\sp\bot} \Sb{\Sigma \medspace\mathord\cup\medspace \Pi'}, \Delta'\Sp{+}\braces*{\brackets*{\Delta''\Sp{+}, \UseLasy\mdwhtdiamond B\cramped{\sp\bot}}\cramped{_B}}} } }
    \end{prooftree*}
    where $\Pi' \subseteq \Sigma \cup (\Pi \setminus\braces B)$,
    and so $\Sigma \cup \Pi' \subseteq \Sigma \cup (\Pi \setminus\braces B)$ as required.
    Observe that the~$\UseLasy\mdwhtdiamond B\cramped{\sp\bot}$ in the left premise is
    due to the operation~$(\mathord{-})\Sp{+}$ of the induction hypothesis, whereas
    that in the right comes from~\eqref{eq:dag:1}.
  \end{case}
  \begin{case}
    Suppose
    \begin{prooftree*}[right label template=\rlap{\normalfont\inserttext}]
      \begingroup
        \set{rule style=no rule, template=}
        \begingroup
          \set{rule margin=0pt}
          \hypo{}
          \ellipsis{\small$\mathscr{D}$}{}
        \endgroup
      \endgroup
      \infer[no rule]1{\Gamma\braces*{\brackets[\big]{\UseLasy\mdwhtdiamond A\cramped{\sp\bot} \Sb{\Pi}, \Delta'\braces*{\brackets*{\UseLasy\mdwhtdiamond C\cramped{\sp\bot}, C}\cramped{_C}}}}}
      \infer1[($\UseLasy\mdsmwhtsquare$)] { \Gamma\braces*{\brackets[\big]{\UseLasy\mdwhtdiamond A\cramped{\sp\bot} \Sb{\Pi}, \Delta'\braces*{\UseLasy\mdsmwhtsquare C}}} }
    \end{prooftree*}
    Then we have
    \begin{prooftree*}[right label template=\rlap{\normalfont\inserttext}]
      \begingroup
        \set{rule style=no rule, template=}
        \begingroup
          \set{rule margin=0pt}
          \hypo{}
          \ellipsis{\small$\mathscr{D}$}{}
        \endgroup
      \endgroup
      \infer[no rule]1{\Gamma\braces*{\brackets[\big]{\UseLasy\mdwhtdiamond A\cramped{\sp\bot} \Sb{\Pi}, \Delta'\braces*{\brackets*{\UseLasy\mdwhtdiamond C\cramped{\sp\bot}, C}\cramped{_C}}}}}
      \infer1[(IH)] { \Gamma\braces*{\brackets[\big]{\UseLasy\mdwhtdiamond A\cramped{\sp\bot} \Sb{\Pi'}, \Delta'\Sp{+}\braces*{\brackets*{\UseLasy\mdwhtdiamond C\cramped{\sp\bot}, \UseLasy\mdwhtdiamond C\cramped{\sp\bot}, C}\cramped{_C}}}} }
      \infer[dashed]1[(contract)] { \Gamma\braces*{\brackets[\big]{\UseLasy\mdwhtdiamond A\cramped{\sp\bot} \Sb{\Pi'}, \Delta'\Sp{+}\braces*{\brackets*{\UseLasy\mdwhtdiamond C\cramped{\sp\bot}, C}\cramped{_C}}}} }
      \infer1[($\UseLasy\mdsmwhtsquare$)] { \Gamma\braces*{\brackets[\big]{\UseLasy\mdwhtdiamond A\cramped{\sp\bot} \Sb{\Pi'}, \Delta'\Sp{+}\braces*{\UseLasy\mdsmwhtsquare C}}} }
    \end{prooftree*}
    for some $\Pi' \subseteq \Sigma \cup (\Pi \setminus\braces B)$.
    We notice that when $C \equiv B$ here, we have successfully canceled out
    that we added~$\UseLasy\mdwhtdiamond B\cramped{\sp\bot}$ in \cref{claim:step:1st}.
  \end{case}
  The other cases are straightforward.
\end{proof}

Repeated application of these two \lcnamecrefs{claim:step:2nd}
sweeps all the annotations away from~$\UseLasy\mdwhtdiamond A\cramped{\sp\bot}$:

\begin{lemma}\label{claim:annotation-elim}
  If
  \[
    \mathord\vdash\sp{A}\penalty\binoppenalty\mskip\thickmuskip \Gamma\braces*{\brackets*{\UseLasy\mdwhtdiamond A\cramped{\sp\bot} \sb{\Sigma}, A}} \mkern1mu\text,
    \tag{\ddag} \label{eq:ddag}
  \]
  then\/ $\mathord\vdash\Sp{A}\penalty\binoppenalty\mskip\thickmuskip \Gamma\braces{\brackets{\UseLasy\mdwhtdiamond A\cramped{\sp\bot} \Sb\emptyset, A}}$.
\end{lemma}

\begin{proof}
  By induction on the size of $\Sigma$.
  If $\Sigma = \emptyset$, the proof is complete; otherwise, let $B \in \Sigma$.
  From~\eqref{eq:ddag}, \cref{claim:step:1st} yields
  $\mathord\vdash\Sp{A}\penalty\binoppenalty\mskip\thickmuskip \Gamma\braces{\brackets{\UseLasy\mdwhtdiamond A\cramped{\sp\bot} \Sb{\Sigma'} \mathcomma\penalty\binoppenalty \UseLasy\mdwhtdiamond B\cramped{\sp\bot} \mathcomma\penalty\binoppenalty A}}$
  for some $\Sigma' \subseteq \Sigma \setminus\braces B$.
  Then putting this into~\eqref{eq:dag:1}
  and~\eqref{eq:ddag} into~\eqref{eq:dag:2} in \cref{claim:step:2nd},
  we have $\mathord\vdash\Sp{A}\penalty\binoppenalty\mskip\thickmuskip \Gamma\braces{\brackets{\UseLasy\mdwhtdiamond A\cramped{\sp\bot} \Sb{\Sigma''}, A}}$ for some
  $\Sigma'' \subseteq \Sigma' \cup (\Sigma \setminus\braces B) \subsetneq \Sigma$.
  Applying the induction hypothesis to~$\Sigma''$, we obtain the conclusion.
\end{proof}

Now we may drop unused assumptions from a derived sequent, which
can be restated formally in a specific form we need as follows:

\begin{lemma}[Thinning]\label{claim:thinning}
  If\/ $\mathord\vdash\Sp{A}\penalty\binoppenalty\mskip\thickmuskip \Gamma\braces{\UseLasy\mdwhtdiamond A\cramped{\sp\bot} \Sb\emptyset}$, then\/ $\mathord\vdash\Sp{A}\penalty\binoppenalty\mskip\thickmuskip \Gamma\braces{}$.
\end{lemma}

\begin{proof}
  By induction on derivation.
\end{proof}

\Cref{claim:main} follows immediately
from \cref{claim:annotation-elim,claim:thinning}.

\section{Syntactic Cut-Elimination}\label{sec:cut-elimination}

We are now ready to prove the cut-elimination theorem.
We use the standard double induction
to show the reduction \lcnamecref{claim:reduction}:

\begin{lemma}[Reduction]\label{claim:reduction}
  Suppose
  \begin{prooftree*}[right label template=\rlap{\normalfont\inserttext}]
    \begingroup
      \set{rule style=no rule, template=}
      \begingroup
        \set{rule margin=0pt}
        \hypo{}
        \ellipsis{\small$\mathscr{D}_1$}{}
      \endgroup
    \endgroup
    \infer[no rule]1{ \Gamma\braces*{A} }
    \begingroup
      \set{rule style=no rule, template=}
      \begingroup
        \set{rule margin=0pt}
        \hypo{}
        \ellipsis{\small$\mathscr{D}_2$}{}
      \endgroup
    \endgroup
    \infer[no rule]1{ \Gamma\braces*{A\cramped{\sp\bot}} }
    \infer2[(cut)]  { \Gamma\braces*{} }
  \end{prooftree*}
  If~$\mathscr{D}_1$ and~$\mathscr{D}_2$ are both cut-free derivations,
  then we have\/ $\mathord\vdash\penalty\binoppenalty\mskip\thickmuskip \Gamma\braces{}$.
\end{lemma}

\begin{proof}
  We proceed by induction on the following lexicographic ordering:
  \begin{enumerate}[(i)]
    \item The size of~$A$; and\samepage
    \item The sum of the heights of~$\mathscr{D}_1$ and~$\mathscr{D}_2$
  \end{enumerate}
  and reduce the cut to lower ones.

  We here show only the case in question:
  \[
    \begin{prooftree}
      \begingroup
        \set{rule style=no rule, template=}
        \begingroup
          \set{rule margin=0pt}
          \hypo{}
          \ellipsis{\small$\mathscr{D}_1$}{}
        \endgroup
      \endgroup
      \infer[no rule]1{ \Gamma\braces*{\Delta\braces*{A\cramped{\sp\bot}}, \UseLasy\mdwhtdiamond A\cramped{\sp\bot}} }
      \infer1[($\UseLasy\mdwhtdiamond$)] { \Gamma\braces*{\Delta\braces*{}, \UseLasy\mdwhtdiamond A\cramped{\sp\bot}} }
      \set{right label template=\rlap{\normalfont\inserttext}}
      \begingroup
        \set{rule style=no rule, template=}
        \begingroup
          \set{rule margin=0pt}
          \hypo{}
          \ellipsis{\small$\mathscr{D}_2$}{}
        \endgroup
      \endgroup
      \infer[no rule]1 { \Gamma\braces*{\Delta\braces*{}, \left\lbrack \UseLasy\mdwhtdiamond A\cramped{\sp\bot}, A \right\rbrack} }
      \infer1[($\UseLasy\mdsmwhtsquare$)] { \Gamma\braces*{\Delta\braces*{}, \UseLasy\mdsmwhtsquare A} }
      \infer2[(cut)] { \Gamma\braces*{\Delta\braces*{}} }
    \end{prooftree}
  \]
  By \cref{claim:main} we have $\mathord\vdash\Sp{A}\penalty\binoppenalty\mskip\thickmuskip \Gamma\braces{\Delta\braces{}, \brackets{A}}$,
  where all cuts are admissible by the induction hypothesis, and
  hence $\mathord\vdash\penalty\binoppenalty\mskip\thickmuskip \Gamma\braces{\Delta\braces{}, \brackets{A}}$.
  We can now reduce the cut above as follows:
  \begin{prooftree*}
    \begingroup
      \set{rule style=no rule, template=}
      \begingroup
        \set{rule margin=0pt}
        \hypo{}
        \ellipsis{\small$\mathscr{D}_1$}{}
      \endgroup
    \endgroup
    \infer[no rule]1 { \Gamma\braces*{\Delta\braces*{A\cramped{\sp\bot}}, \UseLasy\mdwhtdiamond A\cramped{\sp\bot}} }
    \begingroup
      \set{rule style=no rule, template=}
      \begingroup
        \set{rule margin=0pt}
        \hypo{}
        \ellipsis{\small$\mathscr{D}_2$}{}
      \endgroup
    \endgroup
    \infer[no rule]1 { \Gamma\braces*{\Delta\braces*{}, \left\lbrack \UseLasy\mdwhtdiamond A\cramped{\sp\bot}, A \right\rbrack} }
    \infer1[($\UseLasy\mdsmwhtsquare$)] { \Gamma\braces*{\Delta\braces*{}, \UseLasy\mdsmwhtsquare A} }
    \infer[dashed]1[(weak)] { \Gamma\braces*{\Delta\braces*{A\cramped{\sp\bot}}, \UseLasy\mdsmwhtsquare A} }
    \infer2[(cut)]  { \Gamma\braces*{\Delta\braces*{A\cramped{\sp\bot}}} }
    \begingroup
      \set{rule style=no rule, template=}
      \begingroup
        \set{rule margin=0pt}
        \hypo{}
        \ellipsis{}{}
      \endgroup
    \endgroup
    \infer[no rule]1       { \Gamma\braces*{\Delta\braces*{}, \brackets*{A}} }
    \infer[dashed]1[(rebase)] { \Gamma\braces*{\Delta\braces*{A}} }
    \infer2[(cut)]    { \Gamma\braces*{\Delta\braces*{}} }
  \end{prooftree*}
  Both of the cuts here are admissible by the induction hypothesis and,
  unlike the reduction~\eqref{reduction:beta}, a third cut is no longer required
  thanks to the diagonal-formula-elimination subprocedure.

  The other cases are standard and almost the same as for the~$\logic K$ case
  of Brünnler~\cite{brünnler2009deep}, where no special consideration
  regarding well-foundedness is needed since we add no new induction parameters.
\end{proof}

\begin{theorem}[Cut-elimination]
  The cut-rule is admissible.
\end{theorem}

\begin{proof}
  Immediate from \cref{claim:reduction}.
\end{proof}

\section{Conclusion}\label{sec:conclusion}

We have presented a syntactic cut-elimination proof for $\logic{GL}$
by combining the following two ideas:
\begin{itemize}
  \item
    The diagonal-formula-elimination subprocedure splits off
    the difficult part of $\logic{GL}$'s cut-elimination,
    thereby allowing for the proof in a more modular way
    without any trouble on the termination of the entire procedure.
  \item
    The nested-sequent approach enables straightforward induction proofs,
    where we rely only on local assumptions and can perform rewriting
    without having to grasp the entire derivation with the help of annotations.
\end{itemize}
This allows for a more concise and clear proof than previous methods
in a composable way.

We employed a context-sharing form of the cut-rule in this paper,
but other forms can be considered.
For example, Poggiolesi~\cite{poggiolesi2009purely} adopted a context-independent one,
and Brünnler~\cite{brünnler2009deep} considered a special form of multicut
called \emph{Y-cut} for modal logics with the axiom~(4).
These variants have some impact on cut-elimination, but do not seem to
provide a fundamental solution in the $\logic{GL}$ case.
Another possibility is to extend cut to be applicable to subsequents.
Such a generalization has been developed for basic modal logics
such as $\logic K$ and $\logic{S4}$ by Chaudhuri, Marin, and Straßburger~\cite{chaudhuri+2016focused}, and
seems to fit also well with $\logic{GL}$, which is future work.

\paragraph*{Acknowledgments.}

We would like to thank the anonymous reviewers
for their insightful comments and suggestions.

\bibliographystyle{eptcs}
\bibliography{main}

\end{document}